\documentclass{llncs}

\usepackage{epsfig}
\usepackage{gensymb}
\usepackage[english]{babel}
\usepackage{graphicx}
\usepackage{amssymb}
\usepackage{multirow}

\title{Optimal Morphs of Convex Drawings\thanks{Work partially supported by MIUR project AMANDA ``Algorithmics for MAssive and Networked DAta'', prot. 2012C4E3KT\_001, and by NSERC of Canada.}}

\date{}

\author{Patrizio Angelini$^1$, Giordano Da Lozzo$^1$, Fabrizio Frati$^1$,\\ Anna Lubiw$^2$, Maurizio Patrignani$^1$, Vincenzo Roselli$^1$
\institute{
$1$ Dipartimento di Ingegneria, Roma Tre University, Italy\\
\email{\{angelini,dalozzo,frati,patrigna,roselli\}@dia.uniroma3.it}\\
$2$  Cheriton School of Computer Science, University of Waterloo, Canada\\
\email{alubiw@uwaterloo.ca}
}}

\newcommand{\morph}[1]{\ensuremath{\langle #1 \rangle}}

\newcommand{\remove}[1]{}

\newtheorem{claimx}{\bf Claim}

\renewenvironment{proof}
{{\bf Proof:}}{\hspace*{\fill}$\Box$\par\vspace{2mm}}

\begin{document}
\pagestyle{plain}

\maketitle

\begin{abstract}
We give an algorithm to compute a morph between any two convex drawings of the same plane graph. The morph preserves the convexity of the drawing at any time instant and moves each vertex along a piecewise linear curve with linear complexity. The linear bound is asymptotically optimal in the worst case. \end{abstract}

\section{Introduction} \label{se:introduction}


Convex drawings of plane graphs are a classical topic of investigation in geometric graph theory. A characterization~\cite{t-prg-84} of the plane graphs that admit convex drawings and a linear-time algorithm~\cite{cyn-lacdp-84} to test whether a graph admits a convex drawing are known. Convex drawings in small area~\cite{br-csc3c-06,bfm-cdcpg-07,ck-cgd3pg-97}, orthogonal convex drawings~\cite{rnn-rgdpg-98,rng-rdpg-04,t-prg-84}, and convex drawings satisfying a variety of further geometric constraints~\cite{hn-cdhpgcpg-10,hn-ltascditpg-10} have also been studied. It is intuitive, but far from trivial to prove, that the space of the
 convex drawings of any $n$-vertex plane graph $G$ is connected; i.e., the points in $\mathbb{R}^{2n}$, each corresponding to the two-dimensional coordinates of a convex drawing of $G$, form a connected set. Expressed in yet another way, 
 there exists a {\em convex morph} between any two convex drawings $\Gamma_s$ and $\Gamma_t$ of the same plane graph $G$, that is, a continuous deformation from $\Gamma_s$ to $\Gamma_t$ so that the intermediate drawing of $G$ is convex at any instant of the deformation. The main result of this paper is the existence of a convex morph between any two convex drawings of the same plane graph such that each vertex moves along a piecewise linear curve with linear complexity during the deformation.

The existence of a convex morph between any two convex drawings of the same plane graph was first proved by Thomassen~\cite{t-dpg-83} more than 30 years ago. Thomassen's result confirmed a conjecture of Gr\"unbaum and Shepard~\cite{gs-tgopg-81} and improved upon a previous result of Cairns~\cite{c-dprc-44}, stating that there exists a continuous deformation, called a {\em morph}, between any two straight-line planar drawings of the same plane graph $G$ such that any intermediate straight-line drawing of $G$ is planar. More recently, motivated by applications in computer graphics, animation, and modeling, a number of algorithms for morphing graph drawings have been designed~\cite{ekp-ifmpg-03,fe-gdm-02,gs-gifpm-01,sg-cmcpt-01,sg-imct-03}. These algorithms aim to construct morphs that preserve the topology of the given drawings at any time, while guaranteeing that the trajectories of the vertices are ``nice'' curves.

%

Straight-line segments are undoubtedly the most readable and appealing curves for the vertex trajectories. However, {\em linear morphs}~--~morphs in which the vertices move along straight lines~--~do not always exist~\cite{ekp-ifmpg-03}. A natural way to overcome this problem is to allow vertices to move along piecewise linear curves.
Since trajectories of large complexity would have a dramatically detrimental impact on the readability of the morph, an important goal is to minimize the complexity of these curves. This problem is formalized as follows. Let $\Gamma_s$ and $\Gamma_t$ be two planar straight-line drawings of a plane graph $G$. Find a sequence $\Gamma_s=\Gamma_1,\dots,\Gamma_k=\Gamma_t$ of planar straight-line drawings of $G$ such that, for $1\leq i\leq k-1$, the linear morph transforming $\Gamma_i$ into $\Gamma_{i+1}$, called a {\em morphing step}, is planar and $k$ is small.

The first polynomial upper bound for this problem was recently obtained by Alamdari {\em et al.}~\cite{aac-mpgdpns-13}. The authors proved that a morph between any two planar straight-line drawings of the same $n$-vertex connected plane graph exists with $O(n^4)$ morphing steps. The $O(n^4)$ bound was later improved to $O(n^2)$~\cite{afpr-mpgde-13} and then to a worst-case optimal $O(n)$ bound by Angelini {\em et al.}~\cite{addfpr-mpgdo-14}. The algorithm of Angelini {\em et al.}~\cite{addfpr-mpgdo-14} can be extended to work for disconnected graphs at the expense of an increase in the number of steps to $O(n^{1.5})$~\cite{abcdfm-ccapdg-15}.

In this paper we give an algorithm to construct a convex morph between any two convex drawings of the same $n$-vertex plane graph with $O(n)$ morphing steps. Our algorithm preserves the convexity of the drawing at any time instant and in fact preserves strict convexity, if the given drawings are strictly-convex. The linear bound is tight in the worst case, as can be shown by adapting the lower bound construction of Angelini {\em et al.}~\cite{addfpr-mpgdo-14}. We remark that Thomassen's algorithm~\cite{t-dpg-83} constructs convex morphs with an exponential number of steps. To the best of our knowledge, no other algorithm is known to construct a convex morph between any two convex drawings of the same plane graph.

The outline of our algorithm is simple. 
Let $\Gamma_s$ and $\Gamma_t$ be two convex drawings of the same {\em convex graph} $G$, that is, a plane graph that admits a convex drawing. 
Determine a connected subgraph $G'$ of $G$ such that removing $G'$ from $G$ results in a smaller convex graph $G''$. 
Then $G'$ lies inside one face $f$ of $G''$.
Morph $\Gamma_s$ into a drawing $\Gamma'_s$ of $G$ and morph $\Gamma_t$ into a drawing $\Gamma'_t$ of $G$ such that the cycle of $G$ corresponding to $f$ is delimited by a convex polygon in $\Gamma'_s$ and in $\Gamma'_t$. These morphs consist of one morphing step each. Remove $G'$ from $\Gamma'_s$ and $\Gamma'_t$ to obtain two convex drawings $\Gamma''_s$ and $\Gamma''_t$ of $G''$. Finally, recursively compute a morph between $\Gamma''_s$ and $\Gamma''_t$.  
Since $f$ remains convex throughout the whole morph from $\Gamma''_s$ to $\Gamma''_t$, a morph of $G$ from $\Gamma'_s$ to $\Gamma'_t$ can be obtained from the morph of $G''$ from $\Gamma''_s$ to $\Gamma''_t$ by suitably drawing $G'$ inside $f$ at each intermediate step of such a morph.
The final morph from $\Gamma_s$ to $\Gamma_t$ consists of the 
morph from $\Gamma_s$ to $\Gamma'_s$ followed by the morph from $\Gamma'_s$ to $\Gamma'_t$, and then the reverse of the morph from  $\Gamma_t$ to $\Gamma'_t$. Our algorithm has two main ingredients.

The first ingredient is a structural decomposition of convex graphs that generalizes a well-known structural decomposition of triconnected planar graphs due to Barnette and Gr\"unbaum~\cite{bg-stcc3p-69}. The latter states that any subdivision of a triconnected planar graph contains a path whose removal results in a subdivision of a smaller triconnected planar graph. For convex graphs we can prove a similar theorem which states, roughly speaking, 
that any convex graph contains a path, or three paths incident to the same vertex, whose removal results in a smaller convex graph. 
Our approach is thus based on \emph{removing} a subgraph from the input graph.  This differs from the recent papers on morphing graph drawings~\cite{aac-mpgdpns-13,addfpr-mpgdo-14,afpr-mpgde-13}, where the basic operation is to {\em contract} (i.e.~move arbitrarily close) a vertex to a neighbor.
One of the difficulties of the previous approach was to determine a trajectory for a contracted vertex inside the moving polygon of its neighbors. By removing a subgraph and forcing the newly formed face to be convex, we avoid this difficulty.

The second ingredient is a relationship between {\em unidirectional morphs} and level planar drawings of hierarchical graphs, which allows us to compute the above mentioned morphs between $\Gamma_s$ and $\Gamma'_s$ and between $\Gamma_t$ and $\Gamma'_t$ with one morphing step. This relationship was first observed by Angelini {\em et al.}~\cite{addfpr-mpgdo-14}. However, in order to use it in our setting, we need to prove that every strictly-convex graph admits a {\em strictly-convex} level planar drawing; this strengthens a result of Hong and Nagamochi~\cite{hn-cdhpgcpg-10} and might be of independent interest.

We leave open the question whether any two straight-line drawings of the same plane graph $G$ can be morphed so that every intermediate drawing has polynomial {\em size} (e.g., the ratio between the length of any two edges is polynomial in the size of $G$ during the entire morph). In order to solve this problem positively, our approach seems to be better than previous ones; intuitively, subgraph removals are more suitable than vertex contractions for a morphing algorithm that doesn't blow up the size of the intermediate drawings. Nevertheless, we haven't yet been able to prove that polynomial-size morphs always exist.

%




\section{Definitions and Preliminaries} \label{se:preliminaries}\label{subse:monotonicity}\label{subse:hierarchical-level}

In this section we give some definitions and preliminaries.

{\bf Drawings and Embeddings.} A \emph{straight-line planar drawing} $\Gamma$ of a graph maps vertices to points in the plane and edges to internally disjoint straight-line segments. Drawing $\Gamma$ partitions the plane into topologically connected regions, called  {\em faces}. The bounded faces are \emph{internal} and the unbounded face is the \emph{outer face}. A vertex (an edge) is {\em external} if it is incident to the outer face and {\em internal} otherwise. A vertex $x$ is \emph{convex}, \emph{flat}, or \emph{concave} in an incident face $f$ in $\Gamma$, if the angle at $x$ in $f$ is smaller than, equal to, or larger than $\pi$ radians, respectively.
Drawing $\Gamma$ is {\em convex} ({\em strictly-convex}) if for each vertex $v$ and each face $f$ vertex $v$ is incident to, $v$ is either convex or flat (is convex) in $f$, if $f$ is internal, and $v$ is either concave or flat (is concave) in $f$, if $f$ is the outer face.
A planar drawing determines a clockwise ordering of the edges incident to each vertex. Two planar drawings of a connected planar graph are \emph{equivalent} if they determine the same clockwise orderings and have the same outer face. A \emph{plane embedding} is an equivalence class of planar drawings. A graph with a plane embedding is a \emph{plane graph}. A {\em convex} ({\em strictly-convex}) graph is a plane graph that admits a convex (resp. strictly-convex) drawing with the given plane embedding.

{\bf Subgraphs and Connectivity.}
A subgraph $G'$ of a plane graph $G$ is regarded as a plane graph whose plane embedding is 
obtained from $G$ by removing all the vertices and edges not in $G'$. We denote by $G-e$ (by $G-S$) the plane graph obtained from $G$ by removing an edge $e$ of $G$ (resp.\ a set $S$ of vertices and their incident edges).

We denote by $\deg(G,v)$ the degree of a vertex $v$ in a graph $G$. A graph $G$ is {\em \mbox{biconnected}} ({\em triconnected}) if removing any vertex (resp.\ any two vertices) leaves $G$ connected. A {\em separation pair} in a graph $G$ is a pair of vertices whose removal disconnects $G$. A biconnected plane graph $G$ is {\em internally triconnected} if introducing a new vertex in the outer face of $G$ and connecting it to all the vertices incident to the outer face of $G$ results in a triconnected graph. Thus, internally triconnected plane graphs form a super-class of triconnected plane graphs. A \emph{split component} of a graph $G$ with respect to a separation pair $\{u, v\}$ is either an edge $(u, v)$ or a maximal subgraph $G'$ of $G$ that does not contain edge $(u,v)$, that contains vertices $u$ and $v$, and such that $\{u,v\}$ is not a separation pair of $G'$; we say that $\{u, v\}$ {\em determines} the split components with respect to $\{u, v\}$. For an internally triconnected plane graph $G$, every separation pair $\{u,v\}$ determines two or three split components; further, in the latter case, one of them is an edge $(u,v)$ not incident to the outer face of $G$. 

A \emph{subdivision} $G'$ of a graph $G$ is a graph obtained from $G$ by replacing each edge $(u,v)$ with a path between $u$ and $v$; the internal vertices of this path are called {\em subdivision vertices}. Given a subgraph $H$ of $G$, the subgraph $H'$ of $G'$ {\em corresponding to} $H$ is obtained from $H$ by replacing each edge $(u,v)$ with a path with the same number of vertices as in $G'$.

{\bf Convex Graphs.}
Convex graphs have been 
thoroughly studied, both combinatorially and algorithmically. Most of the known results about convex graphs are stated in the following setting. The input consists of a plane graph $G$ and a convex polygon $P$ representing the cycle $C$ delimiting the outer face of $G$. The problem asks whether $G$ admits a convex drawing in which $C$ is represented by $P$. The known characterizations for this setting imply characterizations and recognition algorithms for the class of convex graphs (with no constraint on the representation of the cycle delimiting the outer face). Quite surprisingly, the literature seems to lack explicit statements of the characterizations in this unconstrained setting. Here we present two theorems, whose proofs can be easily derived from known results~\cite{cyn-lacdp-84,hn-cdhpgcpg-10,t-prg-84}.

\begin{theorem}\label{th:convex-characterization}
A plane graph is convex if and only if it is a subdivision of an internally triconnected plane graph.
\end{theorem}

\begin{theorem}\label{th:strictly-convex-characterization}
A plane graph is strictly-convex if and only if it is a subdivision of an internally triconnected plane graph and every degree-$2$ vertex is external.
\end{theorem}

{\bf Monotonicity.}
A {\em straight arc} $\vec{xy}$ is a straight line segment directed from a point $x$ to a point $y$; $\vec{xy}$ is \emph{monotone} with respect to an oriented straight line $\vec{d}$ if the projection of $x$ on $\vec{d}$ precedes the projection of $y$ on $\vec{d}$ according to the orientation of $\vec{d}$. A path $(u_1, \dots, u_n)$ is \emph{$\vec{d}$-monotone} if $\vec{u_i u_{i+1}}$ is monotone with respect to $\vec{d}$, for $i=1, \dots, n-1$; a polygon $Q$ is \emph{$\vec{d}$-monotone} if it contains two vertices $s$ and $t$ such that the two paths between $s$ and $t$ in $Q$ are both $\vec{d}$-monotone. A path $P$ (a polygon $Q$) is \emph{monotone} if there exists an oriented straight line $\vec{d}$ such that $P$ (resp.\ $Q$) is $\vec{d}$-monotone. We have the following.

\begin{lemma} \label{le:convex-is-monotone}{\bf (Angelini et al.~\cite{addfpr-mpgdo-14})}
Let $Q$ be a convex polygon and $\vec{d}$ be an oriented straight line not orthogonal to any straight line through two vertices of $Q$. Then $Q$ is $\vec{d}$-monotone.
\end{lemma}

\begin{lemma} \label{le:two-polygons-monotone} Let $Q_1$ and $Q_2$ be strictly-convex polygons sharing an edge $e$ and lying on opposite sides of the line through $e$. Let $P_i$ be the path obtained from $Q_i$ by removing edge $e$, for $i=1,2$. The polygon  $Q$ composed of $P_1$ and $P_2$ is monotone; further, an oriented straight-line with respect to which $Q$ is monotone can be obtained by slightly rotating and arbitrarily orienting a line orthogonal to the line through $e$.
\end{lemma}

\begin{proof}
	Let $u$ and $v$ be the end-vertices of $e$. Let $\vec{l}$ be the straight line through $e$, oriented so that arc $\vec{uv}$ has positive projection on $\vec{l}$. Assume w.l.o.g. that $Q_1$ is to the left of $\vec{l}$ and $Q_2$ is to the right of $\vec{l}$ when walking along $\vec{l}$ according to its orientation. Let $\vec{d}$ be an oriented straight line obtained by counter-clockwise rotating $\vec{l}$ by $\pi/2$ radians. Let $\epsilon>0$ be sufficiently small so that: (i) every angle internal to $Q_1$ and $Q_2$ is smaller than $\pi-\epsilon$ radians and (ii) the line obtained by clockwise rotating $\vec{l}$ of $\epsilon$ radians is not parallel to any line through two vertices of $Q$. Condition (i) can be met because of the strict convexity of $Q_1$ and $Q_2$. Let $\vec{d_{\epsilon}}$ be the oriented straight line obtained by clockwise rotating $\vec{d}$ by $\epsilon$ radians. We claim that $Q$ is $\vec{d_{\epsilon}}$-monotone. By Lemma~\ref{le:convex-is-monotone} and since every angle internal to $Q_1$ is smaller than $\pi-\epsilon$ radians, it follows that $Q_1$ is composed of two $\vec{d_{\epsilon}}$-monotone paths connecting $u$ and a vertex $x_1\neq v$. Analogously, $Q_2$ is composed of two $\vec{d_{\epsilon}}$-monotone paths connecting a vertex $x_2\neq u$ and $v$. Then $Q$ is composed of two $\vec{d_{\epsilon}}$-monotone paths connecting $x_2$ and $x_1$, which concludes the proof.
\end{proof}

{\bf Morphing.} A \emph{linear morph} \morph{\Gamma_1,\Gamma_2} between two straight-line planar drawings $\Gamma_1$ and $\Gamma_2$ of a plane graph $G$ moves each vertex at constant speed along a straight line from its position in $\Gamma_1$ to its position in $\Gamma_2$. A linear morph is {\em planar} if no crossing or overlap occurs between any two edges or vertices during the transformation. A linear morph is {\em convex} ({\em strictly-convex}) if it is planar and each face is delimited by a convex  (resp.\ strictly-convex) polygon at any time instant of the morph. A convex linear morph is called a {\em morphing step}. A {\em unidirectional} linear morph~\cite{bhl-mpgdum-13} is a linear morph in which the straight-line trajectories of the vertices are parallel. A \emph{convex morph} (a {\em strictly-convex morph}) \morph{\Gamma_s,\dots,\Gamma_t} between two convex drawings $\Gamma_s$ and $\Gamma_t$ of a plane graph $G$ is a finite sequence of convex (resp.\ strictly-convex) linear morphs that transforms $\Gamma_s$ into $\Gamma_t$. A {\em unidirectional (strictly-) convex morph} is such that each of its morphing steps is unidirectional.

%


\section{Decompositions of Convex Graphs} \label{se:decompositions}

Our morphing algorithm relies on a lemma stating that, roughly speaking, any convex graph has
a ``simple'' subgraph whose removal results in a smaller convex graph.
A similar result is known for a restricted graph class, namely the subdivisions of triconnected planar graphs. 

On the way to proving that every triconnected planar graph is the skeleton of a convex polytope in $\mathbb{R}^3$, Barnette and Gr\"unbaum~\cite{bg-stcc3p-69} proved that every subdivision of a triconnected planar graph $G$ can be decomposed as follows (see also~\cite{s-rc3clt-13}). Starting from $G$, repeatedly remove a path whose internal vertices have degree two in the current graph, until a subdivision of $K_4$ is obtained. 
Barnette and Gr\"unbaum proved that 
there is such a decomposition in which
every intermediate graph is a subdivision of a simple triconnected plane graph.

We now present a lemma that generalizes Barnette and Gr\"unbaum's decomposition technique so that it applies to convex (not necessarily triconnected) graphs. 

\begin{lemma}\label{le:bg-construction-internally}
Let $G$ be a convex graph. There exists a sequence $G_1,\dots,G_{\ell}$ of graphs such that: (i) $G_1=G$; (ii) $G_{\ell}$ is the simple cycle $C$ delimiting the outer face of $G$; (iii) for each $1\leq i\leq \ell$, graph $G_i$ is a subgraph of $G$ and is a subdivision of a simple internally triconnected plane graph $H_i$; and (iv) for each $1\leq i< \ell$, graph $G_{i+1}$ is obtained either:
\begin{itemize}
\item by deleting the edges and the internal vertices of a path $(u_1,u_2,\dots,u_k)$ with $k\geq 2$ from $G_i$, where $u_2,\dots,u_{k-1}$ are degree-$2$ internal vertices of $G_i$; or
\item by deleting a degree-$3$ internal vertex $u$ of $G_i$ as well as the edges and the internal vertices of three paths $P_1$, $P_2$, and $P_3$ connecting $u$ with three vertices of the cycle $C$ delimiting the outer face of $G$, where $P_1$, $P_2$, and $P_3$ are vertex-disjoint except at $u$ and the internal vertices of $P_1$, $P_2$, and $P_3$ are degree-$2$ internal vertices of $G_i$. 
\end{itemize}
\end{lemma}


\begin{proof}
Set $G_1=G$. Suppose that a sequence $G_1,\dots,G_i$ has been determined. If $G_i=G_{\ell}$ is the cycle delimiting the outer face of $G$, then we are done. Otherwise, we distinguish two cases, based on whether $G_i$ is a subdivision of a triconnected plane graph or not.


Suppose first that {\em $G_i$ is a subdivision of a triconnected plane graph $H_i$}. We construct graphs $G_i,\dots,G_{\ell}$ one by one, in reverse order. Throughout the construction, we maintain the following invariant for every $\ell\geq j>i$. Suppose that $H_j$ contains an internal edge $(u,v)$ that is also an edge of $H_i$. Then there exists no path in $H_i$ that connects $u$ and $v$, that is different from edge $(u,v)$, and all of whose internal vertices are not in $H_j$.

\begin{figure}[t]
\begin{center}
\mbox{\includegraphics[scale=0.5]{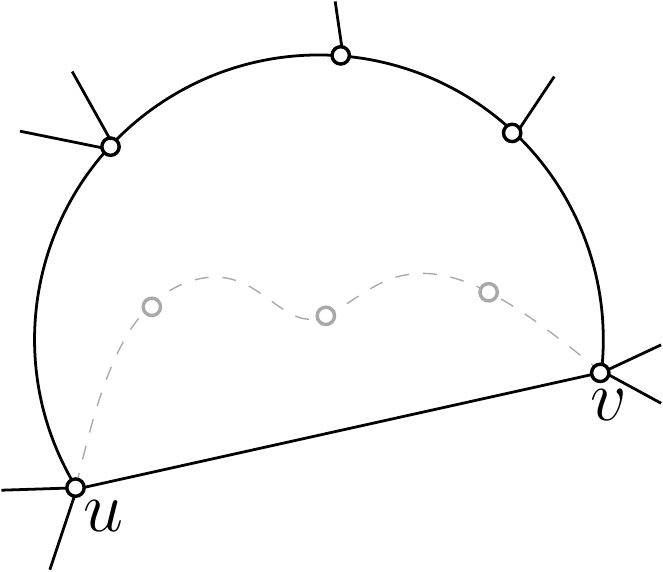}}
\caption{Illustration for the inviariant in the proof of Lemma~\ref{le:bg-construction-internally}. If $H_j$ (solid black lines) contains an internal edge $(u,v)$ that is also an edge of $H_i$, then $H_i$ contains no path that connects $u$ and $v$, that is different from edge $(u,v)$, and all of whose internal vertices are not in $H_j$ (like the gray dashed path in the illustration).}
\label{fig:grunbaum-invariant}
\end{center}
\end{figure}

Let $G_{\ell}$ be the cycle $C$ delimiting the outer face of $G_i$. Next, we determine $G_{\ell-1}$ (see Fig.~\ref{fig:grunbaum-triconnected}(a)). Let $C_i$ be the cycle delimiting the outer face of $H_i$. Since $H_i$ is triconnected and has at least four vertices, there exist three paths that connect an internal vertex $v$ of $H_i$ with vertices of $C_i$, that share no vertices other than $v$, and whose internal vertices are not in $C_i$ (see Theorem~5.1 in~\cite{t-pdfig-80}). Among all the triples of paths with these properties, choose a triple $(P_x,P_y,P_z)$ involving the largest number of vertices of $H_i$. Paths $P_x$, $P_y$, and $P_z$ and cycle $C_i$ form a graph $G_{\ell -1}^H$ that is a subdivision of $K_4$. The subgraph $G_{\ell -1}$ of $G_i$ corresponding to $G_{\ell -1}^H$ is hence a subdivision of $K_4$ in which $v$ is the only degree-$3$ internal vertex. The invariant is satisfied since $P_x$, $P_y$, and $P_z$ involve the largest number of vertices of $H_i$. Further, $G_{\ell}$ is obtained from $G_{\ell-1}$ by deleting a degree-$3$ internal vertex $v$ of $G_{\ell -1}$ as well as the edges and the internal vertices of $P_x$, $P_y$, and $P_z$, as required by the lemma.

\begin{figure}[t]
\begin{center}
\begin{tabular}{c c c c}
\mbox{\includegraphics[scale=0.4]{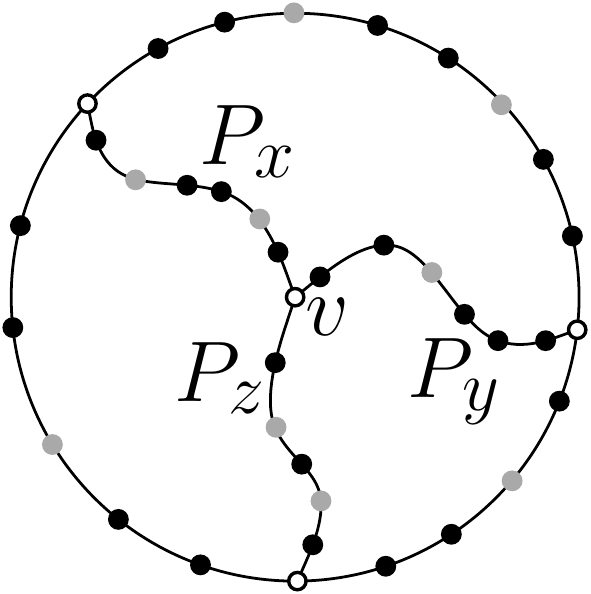}} \hspace{2mm} &
\mbox{\includegraphics[scale=0.4]{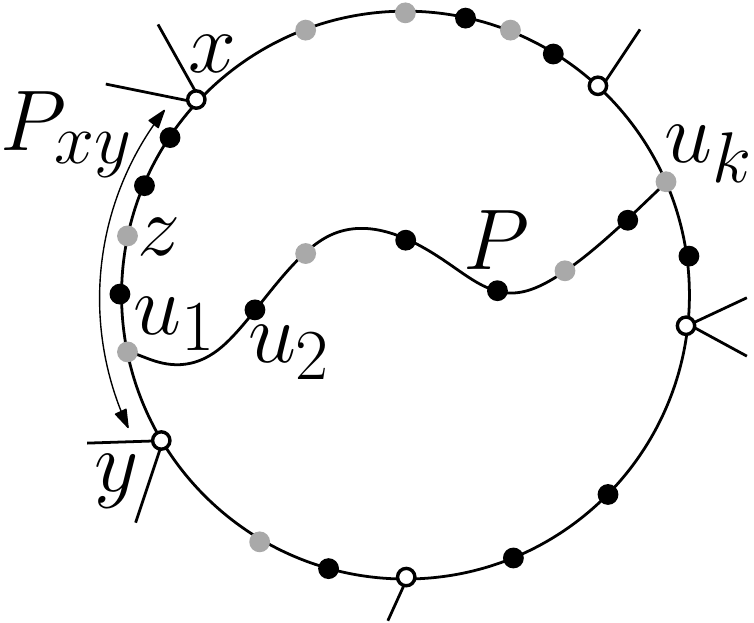}} \hspace{2mm} &
\mbox{\includegraphics[scale=0.4]{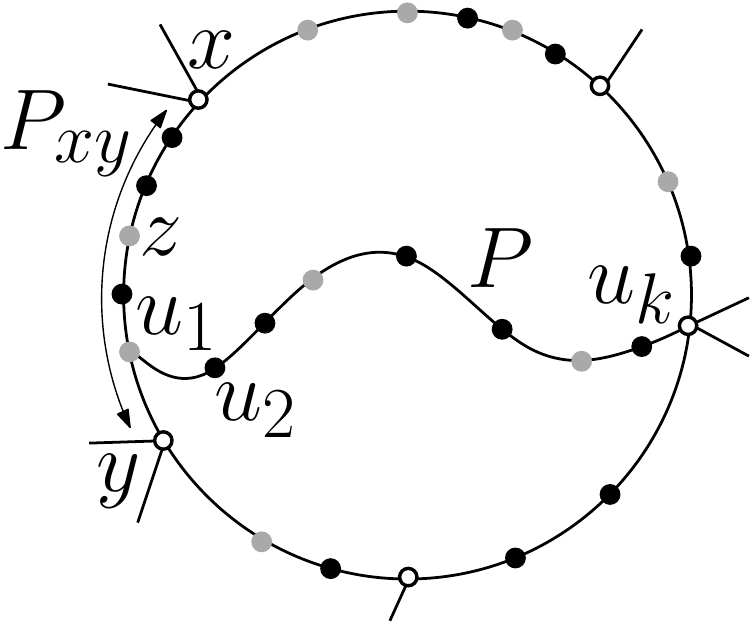}} \hspace{2mm} &
\mbox{\includegraphics[scale=0.4]{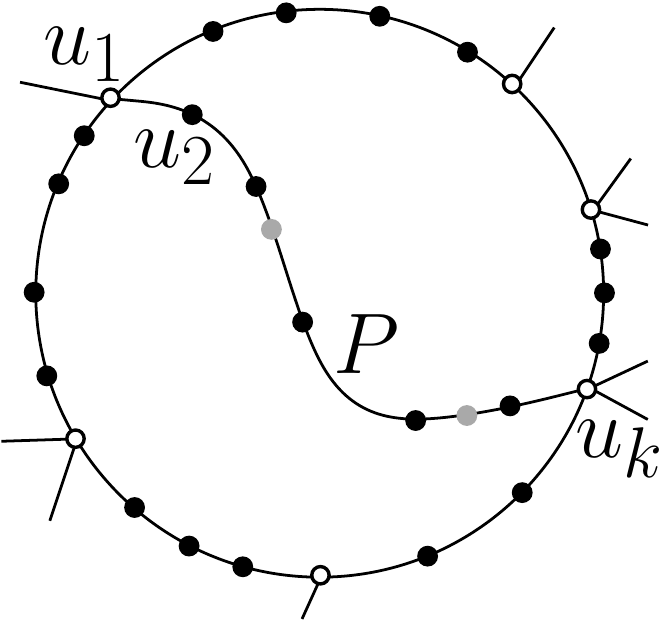}} \\
(a) \hspace{2mm} & (b) \hspace{2mm} & (c) \hspace{2mm} & (d)
\end{tabular}
\caption{Illustration for the proof of Lemma~\ref{le:bg-construction-internally} if $G_i$ is a subdivision of a triconnected plane graph $H_i$. White vertices belong to $H_j$, $G_j$, $H_i$, and $G_i$; grey vertices belong to $G_j$, $H_i$, and $G_i$, and not to $H_j$; black vertices belong to $G_j$ and $G_i$, and not to $H_j$ and $H_i$. (a) Graph $G_{\ell -1}$. (b)--(d) Graph $G_{j}$ and path $P$, together forming graph $G_{j-1}$; (b) and (c) illustrate Case (A) with $u_k$ having degree two and greater than two in $G_j$, respectively, while (d) depicts Case (B).}
\label{fig:grunbaum-triconnected}
\end{center}
\end{figure}

Next, assume that a sequence $G_{\ell},\dots,G_{j}$ has been determined, for some $j\leq \ell-1$. If $G_j=G_i$, then we are done. Otherwise, $G_{j-1}$ is obtained by adding a path $P$ to $G_{j}$. The choice of $P$ distinguishes two cases (as in the proof of Theorem~2 in~\cite{bg-stcc3p-69}).

In Case (A), a vertex $z$ exists such that $\deg(G_{j},z)=2$ and $\deg(G_i,z) \geq 3$. Then, consider the unique path $P_{xy}$ in $G_{j}$ that contains $z$ as an internal vertex, whose internal vertices have degree two in $G_{j}$, and whose end-points $x$ and $y$ have degree at least three in $G_{j}$. Note that $(x,y)$ is an edge of $H_{j}$. Since $\{x,y\}$ is not a separation pair in $H_i$, there exists a path $P=(u_1,u_2,\dots,u_k)$ in $G_i$ such that $u_1$ is an internal vertex of $P_{xy}$, vertex $u_h$ does not belong to $G_j$, for every $2\leq h\leq k-1$, and $u_k$ is a vertex of $G_j$ not in $P_{xy}$. Choose the path with these properties involving the largest number of vertices of $H_i$. Observe that $u_k$ might have degree two (as in Fig.~\ref{fig:grunbaum-triconnected}(b)) or greater than two (as in Fig.~\ref{fig:grunbaum-triconnected}(c)) in $G_j$.

In Case (B), there exists no vertex $z$ such that $\deg(G_{j},z)=2$ and $\deg(G_i,z) \geq 3$ (see Fig.~\ref{fig:grunbaum-triconnected}(d)). Since $G_j$ is different from $G_i$, there exists a path $P=(u_1,u_2,\dots,u_k)$ in $G_i$ such that $u_1$ and $u_k$ belong to $H_j$, and $u_2,\dots,u_{k-1}$ do not belong to $G_j$. Also, a path $P$ satisfying these properties exists such that $u_1$ is an internal vertex of $H_i$ (otherwise $H_i$ would contain a separation pair composed of two external vertices). Choose a path $P$ involving the largest number of vertices of $H_i$, subject to the constraint that $u_1$ is an internal vertex of $H_i$.


In both cases, path $P$ has to be embedded inside a face $f$ of $G_j$, according to the plane embedding of $G_i$. Since $G_j$ contains the cycle delimiting the outer face of $G_i$, we have that $f$ is an internal face of $G_j$. Graph $G_{j-1}$ is obtained by inserting $P$ in $f$. Since $P$ and $G_j$ are subgraphs of $G_i$, graph $G_{j-1}$ is a subgraph of $G_i$. Also, it satisfies the invariant since $P$ is chosen as a path involving the largest number of vertices of $H_i$. It remains to prove that $G_{j-1}$ is a subdivision of a simple triconnected plane graph $H_{j-1}$. Let $H_{j-1}$ be the graph obtained from $G_{j-1}$ by replacing each maximal path whose internal vertices have degree two with a single edge. Thus, $G_{j-1}$ is a subdivision of $H_{j-1}$.

\begin{claimx} \label{cl:decomposition}
Graph $H_{j-1}$ is plane, simple, and triconnected.
\end{claimx}

\begin{proof}
{\em Graph $H_{j-1}$ is a plane graph:} This follows from the fact that $G_{j-1}$ is a plane graph.

{\em Graph $H_{j-1}$ is simple:} In Case (A) graph $H_{j-1}$ is obtained by either: (i) subdividing edge $(x,y)$ of $H_j$ with one subdivision vertex $u_1$ and inserting edge $(u_1,u_k)$, where $u_k$ is a vertex of $H_j$ different from $x$ and $y$; or (ii) subdividing edge $(x,y)$ of $H_j$ with one subdivision vertex $u_1$, subdividing an edge of $H_j$ different from $(x,y)$ with one subdivision vertex $u_k$, and inserting edge $(u_1,u_k)$. Hence, the edges that belong to $H_{j-1}$ and not to $H_j$ are not parallel to any edge of $H_{j-1}$. Then the fact that $H_{j}$ is simple implies that $H_{j-1}$ is simple.  In Case (B) graph $H_{j-1}$ is obtained by adding an edge $e=(u_1,u_k)$ between two vertices of $H_{j}$. Hence, since $H_{j}$ is simple, in order to prove that $H_{j-1}$ is simple as well it suffices to prove that no edge $e'=(u_1,u_k)$ belongs to $H_{j}$. Suppose, for a contradiction, that $H_j$ contains an edge $e'=(u_1,u_k)$. Since there exists no vertex $z$ such that $\deg(G_{j},z)=2$ and $\deg(G_i,z) \geq 3$, it follows that $e'$ is also an edge of $H_i$. Since $H_i$ is simple, $e$ is not an edge of $H_i$. It follows that $e$ corresponds to a path in $H_i$ whose internal vertices do not belong to $H_j$. However, this violates the invariant on $H_j$, a contradiction. 

{\em Graph $H_{j-1}$ is triconnected:} In Case (A) assume that both $u_1$ and $u_k$ are vertices of $H_{j-1}$ not in $H_j$; the case in which $u_k$ is a vertex of $H_j$ is easier to discuss. Suppose, for a contradiction, that graph $H_{j-1}$ contains a separation pair $\{a,b\}$. Denote by $c$ and $d$ vertices in different connected components of $H_{j-1}-\{a,b\}$. If $c=u_1$ and $d=u_k$, then $c$ and $d$ are adjacent in $H_{j-1}-\{a,b\}$, a contradiction. Then assume w.l.o.g. that $d\neq u_1,u_k$. We claim that $c$ is in a connected component of $H_{j-1}-\{a,b\}$ containing a vertex $c'$ of $H_j$. If $c\neq u_1,u_k$, then the claim follows with $c'=c$. Otherwise, assume that $c=u_1$ (the case in which $c=u_k$ is analogous). If $a=x$ and $b=y$, the claim follows with $c'$ being a neighbor of $u_k$ different from $x$ and $y$ (observe that $x$ and $y$ are not both neighbors of $u_k$, otherwise $H_j$ would have two parallel edges $(x,y)$). Finally, if $x\neq a,b$ or if $y\neq a,b$, then the claim follows respectively with $c'=x$ or $c'=y$. Then by the triconnectivity of $H_j$ there exists three paths connecting $c'$ and $d$ in $H_j$ sharing no vertices other than $c'$ and $d$. These paths belong to $H_{j-1}$ as well (one or two edges of these paths might be subdivided in $H_{j-1}$). Hence, at least one of these paths does not contain $a$ nor $b$, thus it belongs to $H_{j-1}-\{a,b\}$. Hence, $c'$ and $d$, and thus $c$ and $d$, are in the same connected component of $H_{j-1}-\{a,b\}$, a contradiction.  In Case (B), $H_{j-1}$ is obtained by adding an edge between two vertices of $H_{j}$. This, together with the fact that $H_{j}$ is triconnected, implies that $H_{j-1}$ is triconnected.
\end{proof}


We now turn to the case in which {\em $G_i$ is not a subdivision of a triconnected plane graph}. In this case $G_i$ is a subdivision of a simple internally triconnected plane graph $H_i$ with minimum degree three and containing some separation pairs. Recall that $H_i$ has either two or three split components with respect to any separation pair $\{u,v\}$.

Suppose that a separation pair $\{u,v\}$ exists in $H_i$ determining three split components. Since $H_i$ is internally triconnected, one of these split components is an internal edge $(u,v)$ of $H_i$ corresponding to a path $P=(u=u_1,\dots,u_k=v)$ in $G_i$, where $u_2,\dots,u_{k-1}$ are degree-$2$ internal vertices of $G_i$. Let $G_{i+1}=G_i-\{u_2,\dots,u_{k-1}\}$ and let $H_{i+1}=H_i-(u,v)$. Note that $G_{i+1}$ is a subdivision of $H_{i+1}$. Then $H_{i+1}$ is an internally triconnected simple plane graph, given that $H_{i}$ is an internally triconnected simple plane graph with three split components with respect to $\{u,v\}$.

Suppose next that every separation pair of $H_i$ determines two split components, as in Fig.~\ref{fig:separation-pairs}(a). Let $\{u,v\}$ be a separation pair of $H_i$ determining two split components $A$ and $B$ such that $A$ does not contain any separation pair of $H_i$ different from $\{u,v\}$, as in Fig.~\ref{fig:separation-pairs}(b), (e.g., let $\{u,v\}$ be a separation pair such that the number of vertices in $A$ is minimum among all separation pairs). Let $L$ be the subgraph of $H_i$ composed of $A$ and of the path $Q$ between $u$ and $v$ that delimits the outer face of $H_i$ and that belongs to $B$; see Fig.~\ref{fig:separation-pairs}(c). Let $D$ be the subgraph of $G_i$ corresponding to $L$; see Fig.~\ref{fig:separation-pairs}(d). The graph $M$ obtained from $L$ by replacing $Q$ with an edge $(u,v)$, shown in Fig.~\ref{fig:separation-pairs}(e), is triconnected, given that the vertex set of $A$ does not contain any separation pair of $H_i$ different from $\{u,v\}$. Thus, $D$ is a subdivision of a simple triconnected plane graph $M$.

\begin{figure}[t]
\begin{center}
\begin{tabular}{c}
\begin{tabular}{c c c c}
\mbox{\includegraphics[scale=0.45]{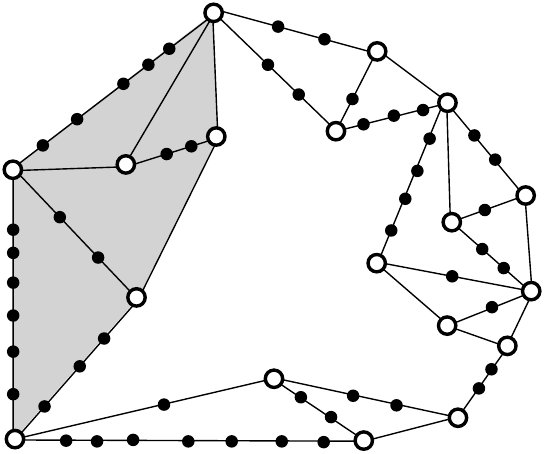}} \hspace{2mm} &
\mbox{\includegraphics[scale=0.45]{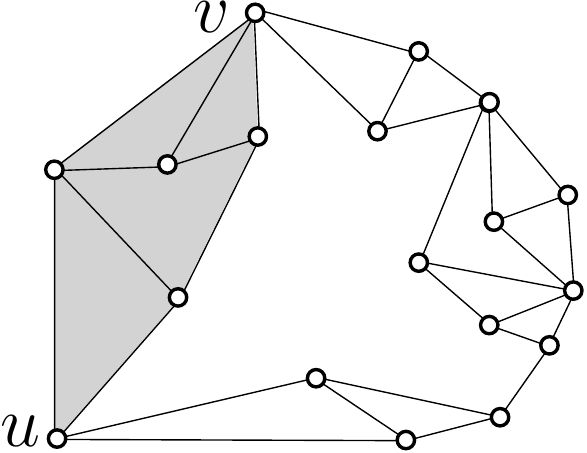}} \hspace{2mm} &
\mbox{\includegraphics[scale=0.45]{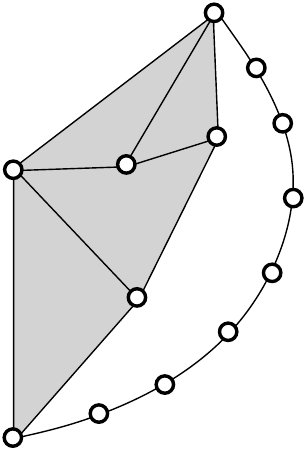}} \hspace{2mm} &
\mbox{\includegraphics[scale=0.45]{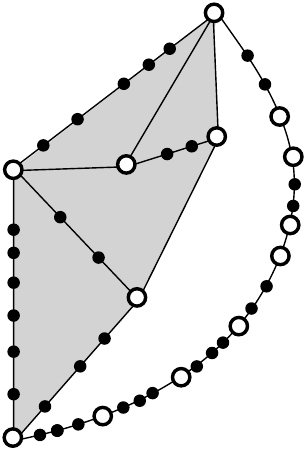}} \\
(a) \hspace{2mm} & (b) \hspace{2mm} & (c) \hspace{2mm} & (d)\vspace{3mm} \\ 
\mbox{\includegraphics[scale=0.45]{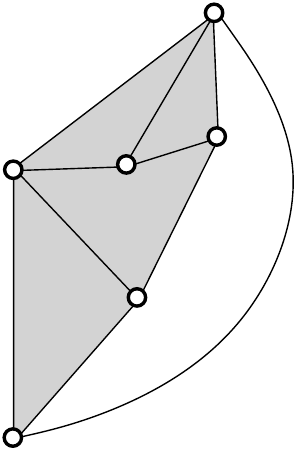}} \hspace{2mm} &
\mbox{\includegraphics[scale=0.45]{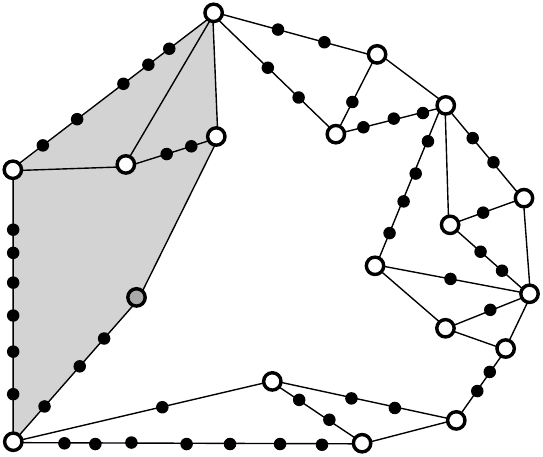}} \hspace{2mm} &
\mbox{\includegraphics[scale=0.45]{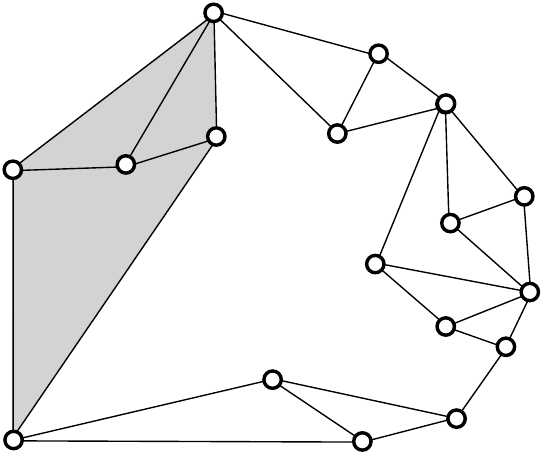}} \hspace{2mm} &
\mbox{\includegraphics[scale=0.45]{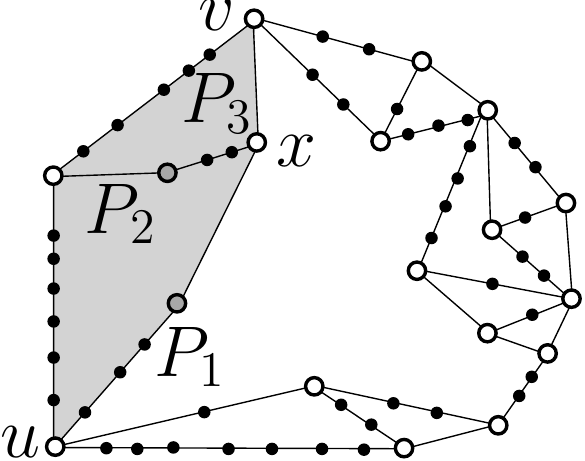}} \\
(e) \hspace{2mm} & (f) \hspace{2mm} & (g) \hspace{2mm} & (h)\\
\vspace{0.5mm}
\end{tabular}\\
\begin{tabular}{c c c}
\mbox{\includegraphics[scale=0.5]{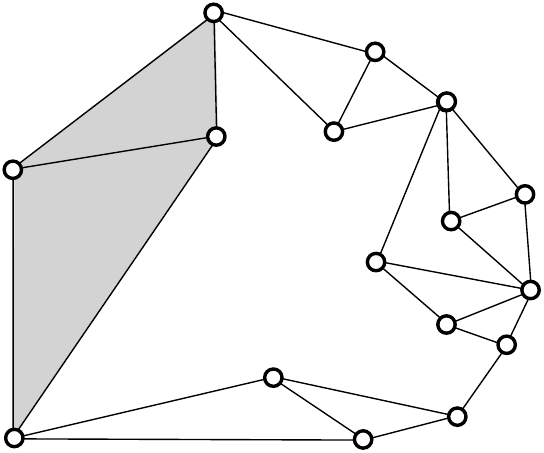}} \hspace{2mm} &
\mbox{\includegraphics[scale=0.5]{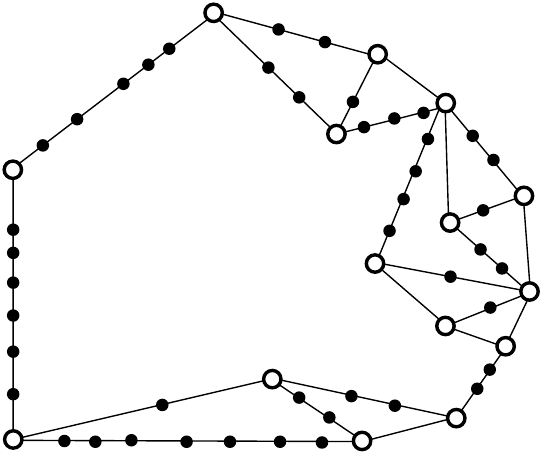}} \hspace{2mm} &
\mbox{\includegraphics[scale=0.5]{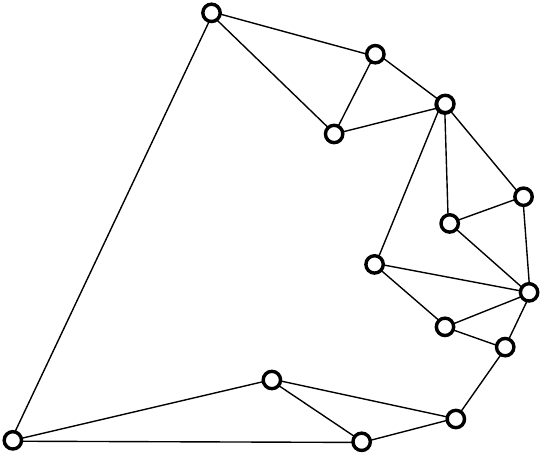}} \\
(i) \hspace{2mm} & (j) \hspace{2mm} & (k)\\
\end{tabular}\\
\end{tabular}
\caption{Illustration for the proof of Lemma~\ref{le:bg-construction-internally} if $G_i$ is not a subdivision of a triconnected plane graph $H_i$. The faces of $D_1,\dots,D_{m}$ not incident to $Q$ are colored gray in $G_i,\dots,G_{i+m-1}$. The faces of $M_1,\dots,M_{m}$ not incident to $(u,v)$ are colored gray in $H_i,\dots,H_{i+m-1}$. (a) Graph $G_i$. (b) Graph $H_i$ and separation pair $\{u,v\}$. (c) Graph $L$. (d) Graph $D=D_1$. (e) Graph $M=M_1$. (f) Graph $G_{i+1}$. (g) Graph $H_{i+1}$. (h) Graph $G_{i+2}$. (i) Graph $H_{i+2}$. (j) Graph $G_{i+3}$. (k) Graph $H_{i+3}$.}
\label{fig:separation-pairs}
\end{center}
\end{figure}

By means of the same algorithm described in the case in which $G_i$ is a subdivision of a triconnected plane graph, we determine a sequence $D_1,\dots,D_m$ of subdivisions of triconnected plane graphs $M_1,\dots,M_m$, where $D_1=D$, $M_1=M$, and $M_m=K_3$. Further, we define a sequence $H_{i+1},\dots,H_{i+m-1}$ of graphs where, for each $2\leq j\leq m-1$, graph $H_{i+j-1}$ is obtained from $H_i$ by replacing $M$ with $M_j$ (see Figs.~\ref{fig:separation-pairs}(b),~\ref{fig:separation-pairs}(g), and~\ref{fig:separation-pairs}(i)), and where $H_{i+m-1}$ is obtained from $H_i$ by replacing $M$ with an edge $(u,v)$ (see Fig.~\ref{fig:separation-pairs}(k)). Analogously, we define a sequence $G_{i+1},\dots,G_{i+m-1}$ of graphs where, for each $2\leq j\leq m$, graph $G_{i+j-1}$ is obtained from $G_i$ by replacing $D$ with $D_j$ (see Figs.~\ref{fig:separation-pairs}(a),~\ref{fig:separation-pairs}(f),~\ref{fig:separation-pairs}(h), and~\ref{fig:separation-pairs}(j)). Then, for each $2\leq j\leq m$, graph $G_{i+j-1}$ is a subdivision of $H_{i+j-1}$.  Further, for each $1\leq j\leq m-2$, graph $G_{i+j}$ is obtained from $G_{i+j-1}$ by deleting the edges and the internal vertices of a path $(u_1,\dots,u_k)$ with $k\geq 2$, where $u_2,\dots,u_{k-1}$ are degree-$2$ internal vertices of $G_{i+j-1}$. Moreover, graph $G_{i+m-1}$ is obtained by deleting from $G_{i+m-2}$ a degree-$3$ internal vertex $x$ as well as the edges and the internal vertices of three paths $P_1$, $P_2$, and $P_3$, as required by the lemma. Finally, since $M_2,\dots,M_m$ are simple triconnected plane graphs, $H_{i+1},\dots,H_{i+m-1}$ are simple internally triconnected plane graphs. 

Note that $H_{i+m-1}$ is obtained from $H_i$ by replacing $A$ with edge $(u,v)$, hence $\{u,v\}$ is not a separation pair in $H_{i+m-1}$. Thus, the repetition of the described transformations over different separation pairs $\{u,v\}$ eventually leads to a graph $G_x$ that is the subdivision of a simple triconnected plane graph $H_x$; then a sequence $G_x,\dots,G_{\ell}$ of subdivisions of triconnected plane graphs such that $G_{\ell}$ is a subdivision of $K_3$ is determined as above.
\end{proof}


\section{Convex Drawings of Hierarchical Convex Graphs} \label{se:hierarchical}

A {\em hierarchical graph} is a tuple $(G,\vec{d},L,\gamma)$ where $G$ is a graph, $\vec{d}$ is an oriented straight line in the plane, $L$ is a set of parallel lines orthogonal to $\vec{d}$, and $\gamma$ is a function that maps each vertex of $G$ to a line in $L$ so that adjacent vertices are mapped to distinct lines. The lines in $L$ are ordered as they are encountered when traversing $\vec{d}$ according to its orientation (we write $l_1<l_2$ if a line $l_1$ precedes a line $l_2$ in $L$). Furthermore, each line $l_i\in L$ is oriented so that $\vec{d}$ cuts $l_i$ from the right to the left of $l_i$; a point $a$ {\em precedes} a point $b$ on $l_i$ if $a$ is encountered before $b$ when traversing $l_i$ according to its orientation. For the sake of readability, we will often write $G$ instead of $(G,\vec{d},L,\gamma)$ to denote a hierarchical graph. A {\em level drawing} of a hierarchical graph $G$ maps each vertex $v$ to a point on the line $\gamma(v)$ and each edge $(u,v)$ of $G$ with $\gamma(u)<\gamma(v)$ to an arc $\vec{uv}$ monotone with respect to $\vec{d}$. 
A hierarchical graph $G$ with a prescribed plane embedding is a {\em hierarchical plane graph} if there is a level planar drawing $\Gamma$ of $G$ that respects the prescribed plane embedding.
A path $(u_1,\dots,u_k)$ in $G$ is \emph{monotone} if $\gamma(u_i)<\gamma(u_{i+1})$, for $1\leq i\leq k-1$. An {\em st-face} in a hierarchical plane graph $G$ is a face delimited by two monotone paths connecting two vertices $s$ and $t$, where $s$ is the {\em source} and $t$ is the {\em sink} of the face. Furthermore, $G$ is a {\em hierarchical-st plane graph} if every face of $G$ is an st-face; note that a face $f$ of $G$ is an st-face if and only if the polygon delimiting $f$ in a straight-line level planar drawing of $G$ is $\vec{d}$-monotone. 

In this section we give an algorithm to construct strictly-convex level planar drawings of {\em hierarchical-st strictly-convex graphs}, that are hierarchical-st plane graphs $(G,\vec{d},L,\gamma)$ such that $G$ is a strictly-convex graph. We have the following.



\begin{theorem} \label{th:hong-nagamochi-revised}
Every hierarchical-st strictly-convex graph admits a drawing which is simultaneously strictly-convex and level planar.
\end{theorem}

\begin{proof}
Let $(G,\vec{d},L,\gamma)$ be a hierarchical-st strictly-convex graph, in the following simply denoted by $G$, and let $C$ be the cycle delimiting the outer face $f$ of $G$. Construct a strictly-convex level planar drawing $P_C$ of $C$ in which the clockwise order of the vertices along $P_C$ is the same as prescribed in $G$. Hong and Nagamochi~\cite{hn-cdhpgcpg-10} showed an algorithm to construct a (non-strictly) convex level planar drawing $\Gamma$ of $G$ in which $C$ is represented by $P_C$. We show how to modify $\Gamma$ into a strictly-convex level planar drawing of $G$.
	
We give some definitions. Let $s$ and $t$ be the vertices of $G$ such that $\gamma(s)<\gamma(u)<\gamma(t)$, for every vertex $u\neq s,t$ of $G$. Given a vertex $v$ of $G$, the {\em leftmost (rightmost) top neighbor} of $v$ is the neighbor $x$ of $v$ with $\gamma(x)>\gamma(v)$ such that for the neighbor $y$ of $v$ counter-clockwise (clockwise) following $x$ we have that either $\gamma(y)<\gamma(v)$, or $\gamma(y)>\gamma(v)$ and both $x$ and $y$ are incident to $f$ (this only happens when $v=s$). The {\em leftmost} and the  {\em rightmost bottom neighbor} of $v$ are defined analogously. Also, the {\em leftmost (rightmost) top path} of $v$ is the monotone path $P$ from $v$ to $t$ obtained by initializing $P=(v)$ and by repeatedly adding the leftmost (resp.\ rightmost) top neighbor of the last vertex. The {\em leftmost} and {\em rightmost bottom path} of $v$ are defined analogously. 
Let $v$ be a vertex of $G$ that is flat in a face $g$ of $\Gamma$; $v$ is an internal vertex of $G$, since  $P_C$ is strictly-convex. Let $x$ and $y$ be the neighbors of $v$ in $g$; then either $\gamma(x) < \gamma(v) < \gamma(y)$ or $\gamma(y) < \gamma(v) < \gamma(x)$. Assume the former. If $g$ lies to the left of path $(x,v,y)$ when traversing it from $x$ to $y$, then we say that $v$ is a {\em left-flat vertex} in $\Gamma$, otherwise $v$ is a {\em right-flat vertex}. By Theorem~\ref{th:strictly-convex-characterization} and since $v$ is an internal vertex of $G$, we have $\deg(G,v)\geq 3$, hence $v$ cannot be both a left-flat and a right-flat vertex in $\Gamma$. A {\em left-flat (right-flat) path} in $\Gamma$ is a maximal path whose internal vertices are all left-flat (resp.\ right-flat) vertices and are all flat in the same face (see Fig.~\ref{fig:strictly-convex}(a)). Let $Q=(x,\dots,y)$ be a left-flat path in $\Gamma$; the {\em elongation} $E_Q$ of $Q$ is the monotone path between $s$ and $t$ obtained by concatenating the rightmost bottom path of $x$, $Q$, and the rightmost top path of $y$. Let $G_l(Q)$ ($G_r(Q)$) be the subgraph of $G$ whose outer face is delimited by the cycle composed of $E_Q$ and of the leftmost (resp. rightmost) top path of $s$. For a right-flat path $Q$ in $\Gamma$, the elongation $E_Q$ of $Q$, and graphs $G_l(Q)$ and $G_r(Q)$ are defined analogously. 

In order to modify $\Gamma$ into a strictly-convex level planar drawing of $G$, we proceed by induction on the number $a(\Gamma)$ of flat angles in $\Gamma$. If $a(\Gamma)=0$, then $\Gamma$ is strictly-convex and there is nothing to be done. If $a(\Gamma) \geq 1$, then there exists a path $Q$ that is either a left-flat path or a right-flat path in $\Gamma$. Assume the former, the other case is symmetric. Also, assume w.l.o.g.\ up to a rotation of the axes, that the lines in $L$ are horizontal. 

Ideally, we would like to move the internal vertices of $Q$ to the right, so that the polygon delimiting the face on which the internal vertices of $Q$ are flat becomes strictly-convex. There is one obstacle to such a modification, though: An internal vertex of $Q$ might be the first or the last vertex of a left-flat path $Q'$; thus, moving that vertex to the right would cause the polygon delimiting the face on which the internal vertices of $Q'$ are flat to become concave (in Fig.~\ref{fig:strictly-convex}(a) moving $u_i$ to the right causes an angle incident to $w_i$ to become concave). We now argue that there is a left-flat path $Q^*$ such that $G_r(Q^*)$ contains no internal left-flat path; then we modify $\Gamma$ by moving the internal vertices of $Q^*$ to the right.

Let $Q^*=(x,\dots,y)$ be a left-flat path such that the number of internal vertices of $G_r(Q^*)$ is minimum. Suppose, for a contradiction, that $G_r(Q^*)$ contains an internal left-flat path $Q'$. Then $G_r(Q')$ has less internal vertices than $G_r(Q^*)$, since $G_r(Q')$ is a subgraph of $G_r(Q^*)$ and the internal vertices of $Q'$ are internal vertices of $G_r(Q^*)$ and external vertices of $G_r(Q')$. This contradiction proves that $G_r(Q^*)$ does not contain any internal left-flat path.


	\begin{figure}[tb]
		\begin{center}
		\begin{tabular}{c c c}
			\mbox{\includegraphics[scale=0.4]{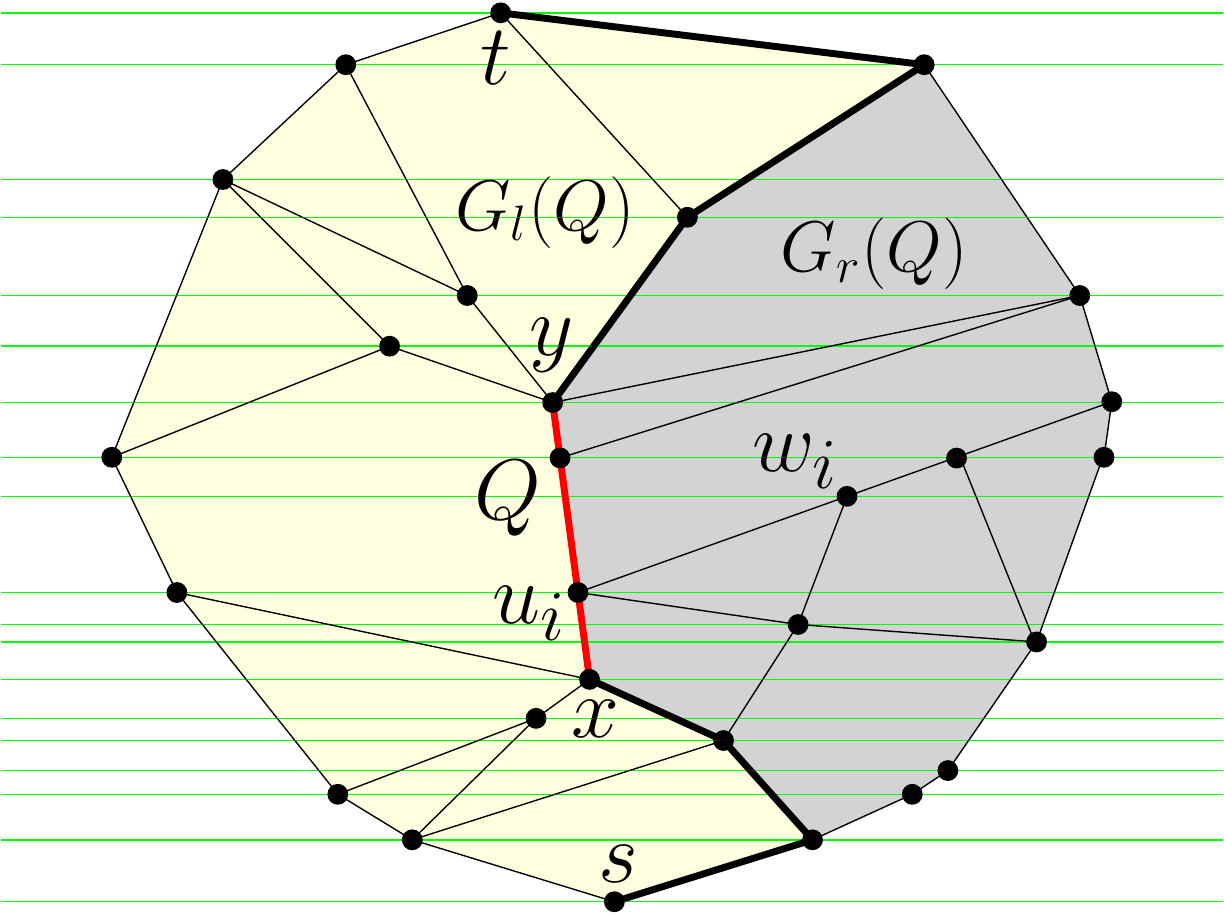}}  \hspace{2mm} &
			\mbox{\includegraphics[scale=0.4]{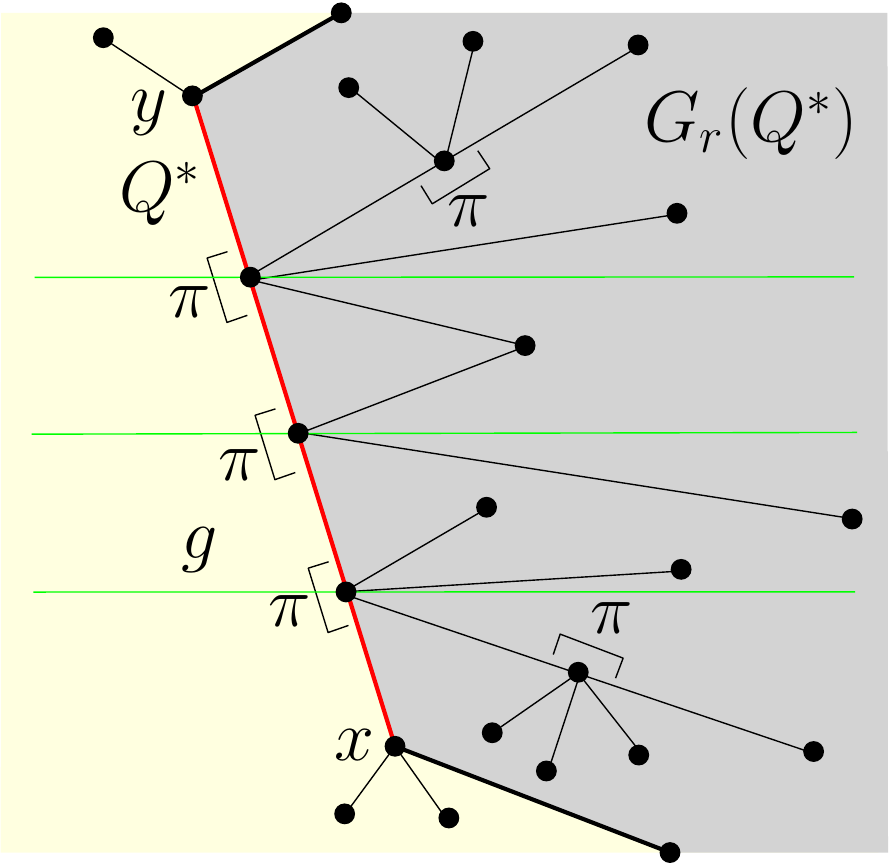}} \hspace{2mm} &
			\mbox{\includegraphics[scale=0.4]{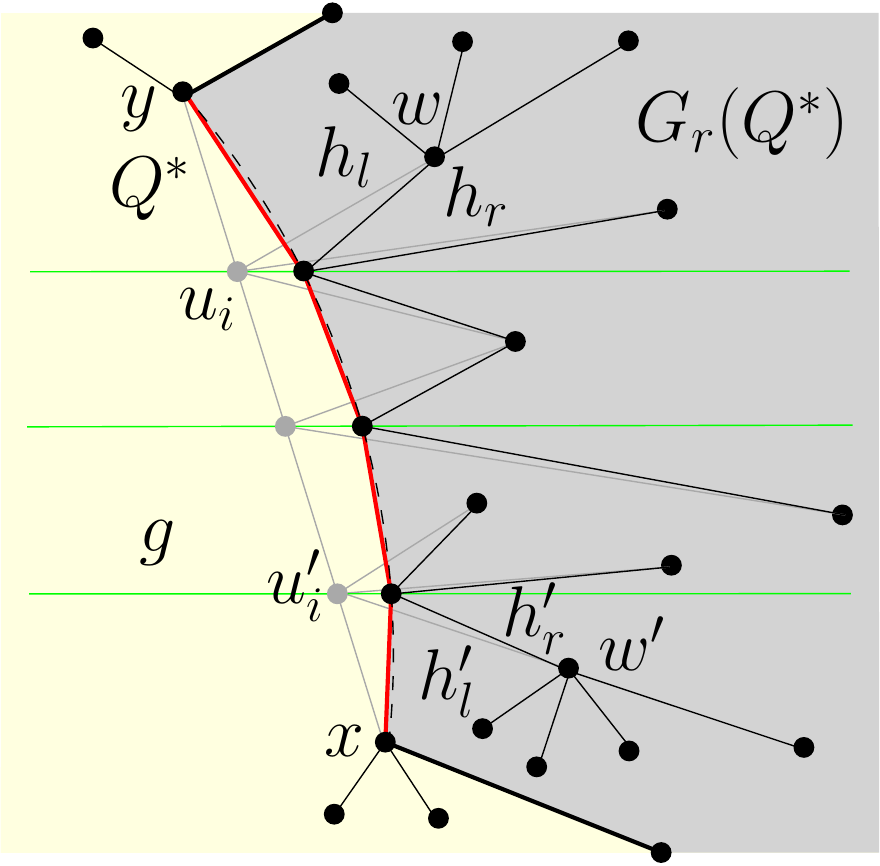}}\\
			(a) \hspace{2mm} & (b) \hspace{2mm} & (c)
		\end{tabular}
		\end{center}
		\caption{(a) A left-flat path $Q$ (red thick line), its elongation $E(Q)$ (red and black thick lines), graphs $G_r(Q)$ (gray) and $G_l(Q)$ (yellow). (b) Drawing $\Gamma$. (c) Drawing $\Gamma'$.}
			\label{fig:strictly-convex}
	\end{figure}
	
	We construct a convex drawing $\Gamma'$ of $G$ with $a(\Gamma')< a(\Gamma)$. Initialize $\Gamma'=\Gamma$ and remove the internal vertices of $Q^*$. Let $\epsilon>0$ be to be determined later. Consider segment $\overline{xy}$, its mid-point $z$, and a point $p$ in the half-plane to the right of $\overline{xy}$ such that segment $\overline{zp}$ is orthogonal to $\overline{xy}$ and has length $\epsilon$. Let $a$ be the arc of circumference between $x$ and $y$ passing through $p$. Place each internal vertex $v$ of $Q^*$ at the intersection point of $\gamma(v)$ with $a$, which exists since $Q^*$ is monotone. Denote by $\Gamma'$ the resulting drawing. We have the following.

\begin{claimx} \label{cl:strictly-claim}
The following statements hold, provided that $\epsilon$ is sufficiently small: (i) $\Gamma'$ is convex; (ii) every vertex that is flat in an incident face in $\Gamma'$ is flat in the same face in $\Gamma$; and (iii) every internal vertex of $Q^*$ is convex in every incident face in $\Gamma'$.
\end{claimx} 

\begin{proof}
For any $\epsilon > 0$, the internal vertices of $Q^*$ are convex in the unique face $g$ of $G_l(Q^*)$ they are all incident to; also, these vertices remain convex in all the internal faces of $G_r(Q^*)$ they are incident to, provided that $\epsilon$ is sufficiently small. Further, since $x$ and $y$ are convex in $g$ in $\Gamma$, they remain convex in $g$ in $\Gamma'$, provided that $\epsilon$ is sufficiently small. Also, for each internal face $h$ of $G_r(Q^*)$ incident to $x$ (to $y$), the angle at $x$ (at $y$) in $h$ in $\Gamma'$ is smaller than or equal to the angle at $x$ (at $y$) in $h$ in $\Gamma$, and hence it is either convex or flat in $\Gamma'$ (if it was flat in $\Gamma$). Finally, consider each edge $(w,u_i) \in G_r(Q^*)$ incident to an internal vertex $u_i$ of $Q^*$. Let $h_l$ and $h_r$ be the two faces incident to $(w,u_i)$ to the left and to the right of $(w,u_i)$, respectively, when traversing $(w,u_i)$ from the vertex with the lowest level to the vertex with the highest level (see Figs.~\ref{fig:strictly-convex}(b)-(c)). The angle at $w$ in $h_r$ is smaller in $\Gamma'$ than in $\Gamma$, hence it is convex in $\Gamma'$ since it was convex of flat in $\Gamma$. Further, the angle at $w$ in $h_l$ is larger in $\Gamma'$ than in $\Gamma$; however, since $G_r(Q^*)$ does not contain any internal left-flat path, it follows that $w$ is not a left-flat vertex in $\Gamma$, hence the angle at $w$ in $h_l$ is convex in $\Gamma$ and it remains convex in $\Gamma'$, provided that $\epsilon$ is sufficiently small. Further, since the positions of all vertices not in $Q^*\setminus \{x,y\}$ are the same in $\Gamma'$ and in $\Gamma$, all the other angles are the same in $\Gamma$ and in $\Gamma'$. 
\end{proof}

Claim~\ref{cl:strictly-claim} implies that $\Gamma'$ is convex and that $a(\Gamma')< a(\Gamma)$. The theorem follows.
\end{proof}


\section{A Morphing Algorithm} \label{se:algorithm}

In this section we give algorithms to morph convex drawings of plane graphs. We start with a lemma about unidirectional linear morphs. Two level planar drawings $\Gamma_1$ and $\Gamma_2$ of a hierarchical plane graph $(G,\vec d,L,\gamma)$ are {\em left-to-right equivalent} if, for any line $l_i\in L$, for any vertex or edge $x$ of $G$, and for any vertex or edge $y$ of $G$, we have that $x$ precedes (follows) $y$ on $l_i$ in $\Gamma_1$ if and only if $x$ precedes (resp.\ follows) $y$ on $l_i$ in $\Gamma_2$. We have the following.

\begin{lemma} \label{le:morphing-equivalent}
	The linear morph \morph{\Gamma_1,\Gamma_2} between two left-to-right equivalent strictly-convex level planar drawings $\Gamma_1$ and $\Gamma_2$ of a hierarchical-st strictly-convex graph $(G,\vec d,L,\gamma)$ is strictly-convex and unidirectional.
\end{lemma} 

\begin{proof}
	It has been proved in~\cite{addfpr-mpgdo-14} that \morph{\Gamma_1,\Gamma_2} is planar and unidirectional. Hence, it suffices to prove that the morph preserves the strict convexity of the drawing at any time instant. 
	In fact, it suffices to prove the following statement. Let $(u,v)$ and $(u,z)$ be any two edges of $G$ such that the clockwise rotation around $u$ that brings $(u,v)$ to coincide with $(u,z)$ is smaller than $\pi$ radians both in $\Gamma_1$ and in $\Gamma_2$; then, at every time instant of the morph, the clockwise rotation around $u$ that brings $(u,v)$ to coincide with $(u,z)$ is smaller than $\pi$ radians.
	
	\begin{figure}[t]
		\begin{center}
			\mbox{\includegraphics[scale=0.35]{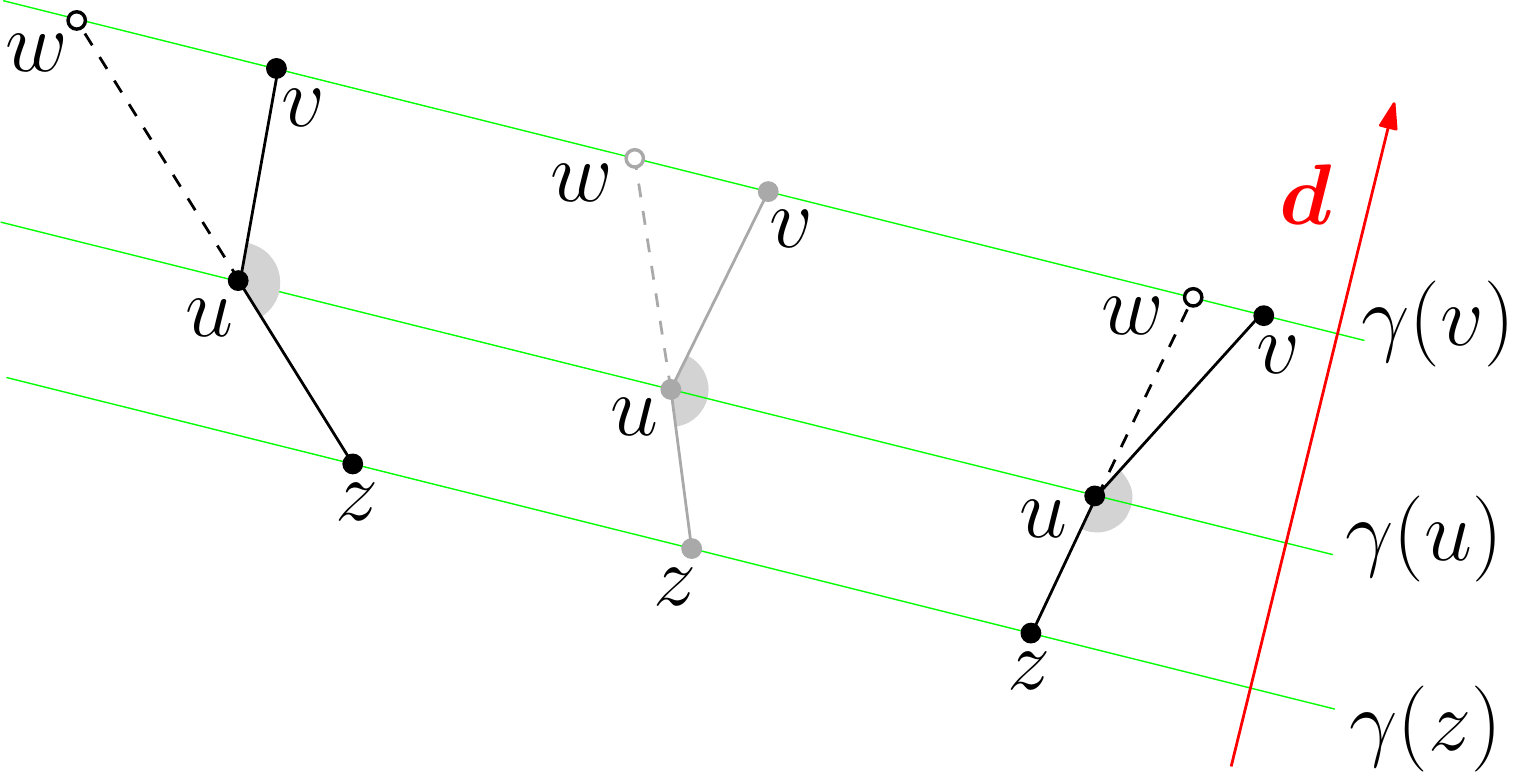}}
			\caption{Illustration for the proof of Lemma~\ref{le:morphing-equivalent}.}
			\label{fig:unidirectional}
		\end{center}
	\end{figure}
	
	The statement descends from the planarity of \morph{\Gamma_1,\Gamma_2} if $\gamma(u)<\gamma(v),\gamma(z)$ or $\gamma(v),\gamma(z)<\gamma(u)$. Otherwise, assume w.l.o.g. that $\gamma(z)<\gamma(u)<\gamma(v)$, as in Fig.~\ref{fig:unidirectional}. Let $w$ denote the intersection point of the line through $u$ and $z$ with $\gamma(v)$. By assumption $w$ precedes $v$ on $\gamma(v)$ both in $\Gamma_1$ and in $\Gamma_2$. Further, $w$ moves at constant speed along $\gamma(v)$ during \morph{\Gamma_1,\Gamma_2}. Then, as proved by Barrera-Cruz {\em et al.}~\cite{bhl-mpgdum-13}, $w$ precedes $v$ on $\gamma(v)$ during the entire morph \morph{\Gamma_1,\Gamma_2}. The statement and hence the lemma follow.
\end{proof}

We now describe an algorithm to construct a strictly-convex morph between any two strictly-convex drawings $\Gamma_s$ and $\Gamma_t$ of a plane graph $G$ with $n$ vertices and $m$ internal faces. The algorithm works by induction on $m$ and consists of at most $2n+2m$ morphing steps. 

{\bf In the base case} we have $m=1$, hence $G$ is a cycle. We have the following.

\begin{claimx}\label{cl:base-case-cycle}
There exists a strictly-convex unidirectional morph with at most $2n+2$ steps between any two strictly-convex drawings $\Gamma_s$ and $\Gamma_t$ of cycle $G$.
\end{claimx}

\begin{proof}
Let $v_1,\dots,v_n$ be the vertices of $G$ in the order they appear when clockwise traversing $G$. Let $\ell$ be a straight line not orthogonal to any line through two vertices of $G$ in $\Gamma_s$ and in $\Gamma_t$. Draw a circumference $\cal C$ enclosing both $\Gamma_s$ and $\Gamma_t$. We morph $\Gamma_s$ ($\Gamma_t$) into a drawing $\Gamma'_s$ ($\Gamma_t'$) such that all the vertices of $G$ are on $\cal C$ (see Fig.~\ref{fig:morphing-base}(a) and Fig.~\ref{fig:morphing-base}(d)) with a single unidirectional strictly-convex morphing step, as follows. Each vertex $v_i$ moves on the line $d_i$ through it orthogonal to $\ell$. Further, each vertex moves in the direction that does not make it collide with the initial drawing of $G$. That is, one of the two half-lines starting at $v_i$ and lying on $d_i$ has no intersection with the only internal face of $G$; then move $v_i$ in the direction defined by that half-line until it hits $\cal C$. This morph is unidirectional (every vertex moves along a line orthogonal to $\ell$) and strictly-convex; the latter follows easily from Lemma~\ref{le:morphing-equivalent}, given that $\Gamma_s$ and $\Gamma'_s$ ($\Gamma_t$ and $\Gamma_t'$) are left-to-right equivalent drawings of the  hierarchical-st strictly-convex graph $(G,\ell,L,\gamma)$, where $L$ consists of the lines $d_i$ and $\gamma$ maps $v_i$ to $d_i$, for $1\leq i\leq n$. 

\begin{figure}[t]
	\begin{center}
		\begin{tabular}{c c c c}
			\mbox{\includegraphics[scale=0.55]{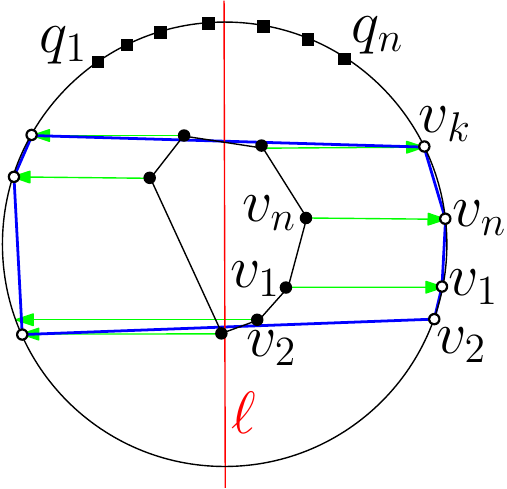}} \hspace{2mm} &
			\mbox{\includegraphics[scale=0.55]{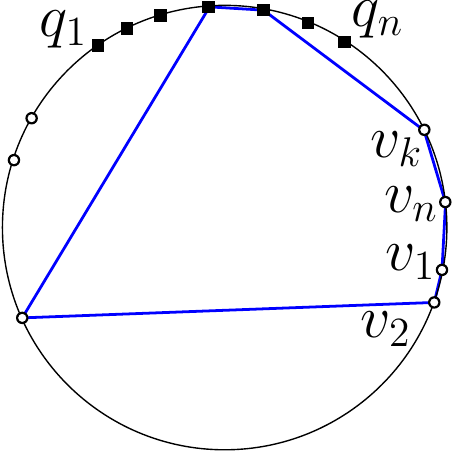}} \hspace{2mm} &
			\mbox{\includegraphics[scale=0.55]{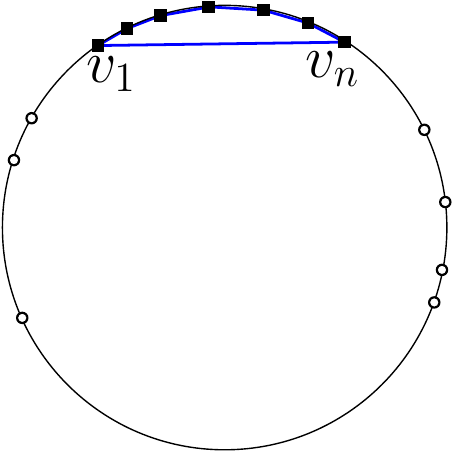}} \hspace{2mm} &
			\mbox{\includegraphics[scale=0.55]{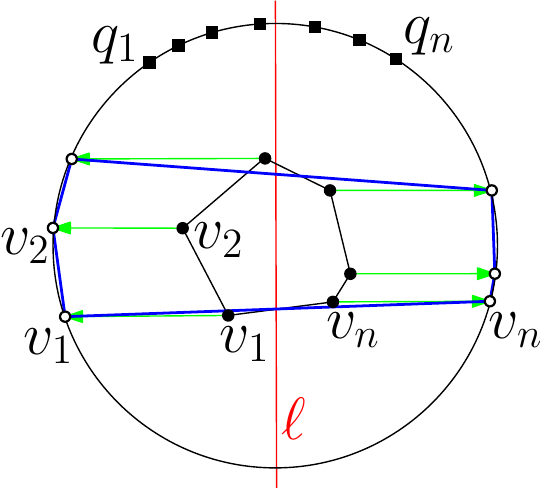}} \\
			(a) \hspace{2mm} & (b) \hspace{2mm} & (c) \hspace{2mm} & (d)
		\end{tabular}
		\caption{(a) Drawings $\Gamma_s$ (black circles and black lines) and $\Gamma'_s$ (white circles and blue lines), together with points $q_1,\dots,q_n$. (b) Morph \morph{\Gamma'_s,\dots,\Gamma} after two steps. (c) Drawing $\Gamma$. (d) Drawings $\Gamma_t$ (black circles and black lines) and $\Gamma'_t$ (white circles and blue lines), together with points $q_1,\dots,q_n$.}
		\label{fig:morphing-base}
	\end{center}
\end{figure}

Consider $n$ points $q_1,\dots,q_n$ in this clockwise order on $\cal C$ both in $\Gamma'_s$ and in $\Gamma'_t$ (see Figs.~\ref{fig:morphing-base}(a) and~\ref{fig:morphing-base}(d)). Points $q_1,\dots,q_n$ are so close that the arc of $\cal C$ between $q_1$ and $q_n$ containing $q_2$ does not contain any vertex of $G$ in $\Gamma'_s$ and $\Gamma'_t$. We morph $\Gamma'_s$ into a drawing $\Gamma$ of $G$ in which $v_i$ is placed at $q_i$, for $1\leq i\leq n$, as follows (see Figs.~\ref{fig:morphing-base}(a)--(c)). Let $v_k$ be the first vertex of $G$ encountered when clockwise traversing $\cal C$ from $q_n$. For $j=k-1,\dots,1,k,\dots,n$, move $v_j$ to $p_j$. Thus morph \morph{\Gamma'_s,\dots,\Gamma} consists of $n$ morphing steps; each morphing step is unidirectional, as a single vertex moves during it, and strictly-convex by Lemma~\ref{le:morphing-equivalent}, since it is a morph between two left-to-right equivalent drawings of the same hierarchical-st plane graph $(G,\vec{d_j},L_j,\gamma_j)$, where $\vec{d_j}$ is an oriented line orthogonal to the line along which $v_j$ moves, $L_j$ consists of the lines orthogonal to $\vec{d_j}$ through $v_1,\dots,v_n$, and $\gamma_j$ maps $v_i$ to the line in $L_j$ through it, for $1\leq i\leq n$. A unidirectional strictly-convex morph \morph{\Gamma'_t,\dots,\Gamma} with $n$ morphing steps is computed analogously. Hence, \morph{\Gamma_s,\Gamma'_s,\dots,\Gamma,\dots,\Gamma'_t,\Gamma_t} is a unidirectional strictly-convex morph between $\Gamma_s$ and $\Gamma_t$ with $2n+2=2n+2m$ morphing steps.
\end{proof}




{\bf In the inductive case} we have $m>1$. Then we apply Lemma~\ref{le:bg-construction-internally} to $G$ in order to obtain a graph $G'$ with $m'<m$ internal faces. We proceed as follows.


Assume first that, according to Lemma~\ref{le:bg-construction-internally}, a degree-$3$ internal vertex $u$ of $G$ as well as the edges and the internal vertices of paths $P_1$, $P_2$, and $P_3$ can be removed from $G$ resulting in a convex graph $G'$, where: (i) $P_1$, $P_2$, and $P_3$ respectively connect $u$ with vertices $u_1$, $u_2$, and $u_3$ of the cycle $C$ delimiting the outer face $f$ of $G$; (ii) $P_1$, $P_2$, and $P_3$ are vertex-disjoint except at $u$; and (iii) the internal vertices of $P_1$, $P_2$, and $P_3$ are degree-$2$ internal vertices of $G$. Graph $G$ has no degree-$2$ internal vertices, since it is strictly-convex (see Theorem~\ref{th:strictly-convex-characterization}), hence $P_1$, $P_2$, and $P_3$ are edges $(u,u_1)$, $(u,u_2)$, and $(u,u_3)$, respectively. 


Vertex $u$ lies in the interior of triangle $\Delta(u_1,u_2,u_3)$ both in $\Gamma_s$ and in $\Gamma_t$, since $\deg(G,u)=3$ and the angles incident to $u$ are smaller than $\pi$ both in $\Gamma_s$ and in $\Gamma_t$. Hence, the position of $u$ is a convex combination of the positions of $u_1$, $u_2$, and $u_3$ both in $\Gamma_s$ and in $\Gamma_t$ (the coefficients of such convex combinations might be different in $\Gamma_s$ and in $\Gamma_t$). Further, no vertex other than $u$ and no edge other than those incident to $u$ lie in the interior of triangle $\Delta(u_1,u_2,u_3)$ in $\Gamma_s$ and $\Gamma_t$, since these drawings are strictly-convex. With a single unidirectional linear morph, move $u$ in $\Gamma_s$ to the point that is a convex combination of the positions of $u_1$, $u_2$, and $u_3$ with the same coefficients as in $\Gamma_t$. This morph is strictly-convex since $u$ stays inside $\Delta(u_1,u_2,u_3)$ at any time instant. Let $\Gamma'_s$ be the resulting drawing of $G$.

Let $Q_1$, $Q_2$, and $Q_3$ be the polygons delimiting the faces of $G$ incident to $u$ in $\Gamma_s$. Let $\Lambda'_s$ be the drawing of $G'$ obtained from $\Gamma'_s$ by removing $u$ and its incident edges. We claim that $\Lambda'_s$ is strictly-convex. Indeed, every internal face of $G'$ different from the face $f_u$ that used to contain $u$ is also a face in $\Gamma'_s$, hence it is delimited by a strictly-convex polygon. Further, every internal angle of the polygon delimiting $f_u$ is either an internal angle of $Q_1$, $Q_2$, or $Q_3$, hence it is smaller than $\pi$, since $\Gamma'_s$ is strictly-convex, or is incident to $u_1$, $u_2$, or $u_3$; however, these vertices are concave in $f$, hence they are convex in $f_u\neq f$. Analogously, the drawing $\Lambda'_t$ of $G'$ obtained from $\Gamma_t$ by removing $u$ and its incident edges is strictly-convex. 

Inductively construct a unidirectional convex morph \morph{\Lambda'_s=\Lambda_0,\dots,\Lambda_\ell=\Lambda'_t} with $\ell\leq 2(n-1)+2(m-2)$ morphing steps. For each $1\leq j\leq \ell-1$, draw $u$ in $\Lambda_j$ at a point that is the convex combination of the positions of $u_1$, $u_2$, and $u_3$ with the same coefficients as in $\Gamma'_s$ and in $\Gamma_t$; denote by $\Gamma_j$ the resulting drawing of $G$. Morph \morph{\Gamma'_s=\Gamma_0,\dots,\Gamma_\ell=\Gamma_t} is strictly-convex and unidirectional. Namely, in every morphing step \morph{\Gamma_j,\Gamma_{j+1}}, vertex $u$ moves between two points that are convex combinations of the positions of $u_1$, $u_2$, and $u_3$ with the same coefficients, hence it moves parallel to each of $u_1$, $u_2$, and $u_3$ (from which \morph{\Gamma_0,\dots,\Gamma_\ell} is unidirectional) and it stays inside $\Delta(u_1,u_2,u_3)$ at any time instant of \morph{\Gamma_j,\Gamma_{j+1}} (from which \morph{\Gamma_0,\dots,\Gamma_\ell} is strictly-convex). Thus, \morph{\Gamma_s,\Gamma'_s=\Gamma_0,\dots,\Gamma_\ell=\Gamma_t} is a unidirectional strictly-convex morph between $\Gamma_s$ and $\Gamma_t$ with $\ell+1\leq 2n+2m-5$ morphing steps.


Assume next that, according to Lemma~\ref{le:bg-construction-internally}, the edges and the internal vertices of a path $P$, whose internal vertices are degree-$2$ internal vertices of $G$, can be deleted from $G$ so that the resulting graph $G'$ is convex. Graph $G$ has no degree-$2$ internal vertices, since it is strictly-convex (see Theorem~\ref{th:strictly-convex-characterization}), hence $P$ is an edge $(u,v)$. Removing $(u,v)$ from $\Gamma_s$ (from $\Gamma_t$) results in a drawing $\Lambda_s$ (resp.\ $\Lambda_t$) of $G'$ which is not, in general, convex, since vertices $u$ and $v$ might be concave in the face $f_{uv}$ of $G'$ that used to contain $(u,v)$, as in Fig.~\ref{fig:case2}. By Lemma~\ref{le:two-polygons-monotone}, there exists an oriented straight line $\vec{d_s}$ such that the polygon $Q_{uv}$ representing the cycle $C_{uv}$ delimiting $f_{uv}$ is $\vec{d_s}$-monotone. By slightly perturbing the slope of $\vec{d_s}$, we can assume that it is not orthogonal to any line through two vertices of $G'$. Let $L'_s$ be the set of parallel and distinct lines through vertices of $G'$ and orthogonal to $\vec{d_s}$. Let $\gamma'_s$ be the function that maps each vertex of $G'$ to the line in $L'_s$ through it. We have the following. 

\begin{lemma} \label{le:orientation}
	$(G',\vec{d_s},L'_s,\gamma'_s)$ is a hierarchical-st convex graph.
\end{lemma}

\begin{proof}
	By construction, $\Lambda_s$ is a straight-line level planar drawing of $(G',\vec{d_s},L'_s,\gamma'_s)$, hence $(G',\vec{d_s},L'_s,\gamma'_s)$ is a hierarchical plane graph. Further, $G'$ is a convex graph by assumption. Moreover, every polygon delimiting a face of $G'$ in $\Lambda_s$ is $\vec{d_s}$-monotone. This is true for $Q_{uv}$ by construction and for every other polygon $Q$ delimiting a face of $G'$ in $\Lambda_s$ by Lemma~\ref{le:convex-is-monotone}, given that $Q$ is a strictly-convex polygon. Since every polygon delimiting a face of $G'$ in $\Lambda_s$ is $\vec{d_s}$-monotone, every face of $G'$ is an st-face, and the lemma follows. 
\end{proof}

Analogously, there exists an oriented straight line $\vec{d_t}$ that leads to define a hierarchical-st convex graph $(G',\vec{d_t},L'_t,\gamma'_t)$ for which $\Lambda_t$ is a straight-line level planar drawing. 

We now distinguish three cases, based on whether $\deg(G',u),\deg(G',v)>2$ {\bf (Case 1)}, $\deg(G',u)=2$ and $\deg(G',v)>2$ {\bf (Case 2)}, or $\deg(G',u)=\deg(G',v)=2$ {\bf (Case 3)}. The case in which $\deg(G',u)>2$ and $\deg(G',v)=2$ is symmetric to Case 2.


{\bf In Case~1} graph $G'$ is strictly-convex, since it is convex and all its internal vertices have degree greater than two. By Theorem~\ref{th:hong-nagamochi-revised}, $(G',\vec{d_s},L'_s,\gamma'_s)$ and $(G',\vec{d_t},L'_t,\gamma'_t)$ admit strictly-convex level planar drawings $\Lambda'_s$ and $\Lambda'_t$, respectively. Let $\Gamma'_s$ ($\Gamma'_t$) be the strictly-convex level planar drawing of $(G,\vec{d_s},L'_s,\gamma'_s)$ (resp.\ of $(G,\vec{d_t},L'_t,\gamma'_t)$) obtained by inserting edge $(u,v)$ as a straight-line segment in $\Lambda'_s$ (resp.\ $\Lambda'_t$). Drawings $\Gamma_s$ and $\Gamma'_s$ ($\Gamma_t$ and $\Gamma'_t$) are left-to-right equivalent. This is argued as follows. First, since $G$ is a plane graph, its outer face is delimited by the same cycle $C$ in both $\Gamma_s$ and $\Gamma'_s$; further, the clockwise order of the vertices along $C$ is the same in $\Gamma_s$ and in $\Gamma'_s$ (recall that Theorem~\ref{th:hong-nagamochi-revised} allows us to arbitrarily prescribe the strictly-convex polygon representing $C$). Consider any two vertices or edges $x$ and $y$ both intersecting a line $\ell$ in $L'_s$; assume this line to be oriented in any way. Suppose, for a contradiction, that $x$ precedes $y$ on $\ell$ in $\Gamma_s$ and follows $y$ on $\ell$ in $\Gamma'_s$. Since $\Gamma_s$ and $\Gamma'_s$ are strictly-convex, there exists a $\vec{d_s}$-monotone path $P_x$ ($P_y$) containing $x$ (resp. $y$) and connecting two vertices of $C$. Then $P_x$ and $P_y$ properly cross, contradicting the planarity of $\Gamma_s$ or of $\Gamma'_s$, or they share a vertex which has a different clockwise order of its incident edges in the two drawings, contradicting the fact that $\Gamma_s$ and $\Gamma'_s$ are drawings of the same plane graph. By Lemma~\ref{le:morphing-equivalent}, linear morphs \morph{\Gamma_s,\Gamma'_s} and \morph{\Gamma_t,\Gamma'_t} are strictly-convex and unidirectional.


%
Inductively construct a unidirectional strictly-convex morph \morph{\Lambda'_s=\Lambda_0,\Lambda_1,\dots,\Lambda_\ell=\Lambda'_t} with $\ell\leq 2n+2(m-1)$ morphing steps between $\Lambda'_s$ and $\Lambda'_t$. For each $0\leq j\leq \ell$, draw edge $(u,v)$ in $\Lambda_j$ as a straight-line segment $\overline{uv}$; let $\Gamma_j$ be the resulting drawing of $G$. We have that morph \morph{\Gamma'_s=\Gamma_0,\Gamma_1,\dots,\Gamma_\ell=\Gamma'_t} is strictly-convex and unidirectional given that \morph{\Lambda_0,\Lambda_1,\dots,\Lambda_\ell} is strictly-convex and  unidirectional and given that, at any time instant of \morph{\Lambda_0,\Lambda_1,\dots,\Lambda_\ell}, segment $\overline{uv}$ splits the strictly-convex polygon delimiting $f_{uv}$ into two strictly-convex polygons. Thus, \morph{\Gamma_s,\Gamma'_s=\Gamma_0,\Gamma_1,\dots,\Gamma_\ell=\Gamma'_t,\Gamma_t} is a unidirectional strictly-convex morph between $\Gamma_s$ and $\Gamma_t$ with $\ell+2\leq 2n+2m$ morphing steps. 


{\bf In Case~2} let $G''$ be the graph obtained from $G'$ by replacing path $(x,u,y)$ with edge $(x,y)$, where $x$ and $y$ are the only neighbors of $u$ in $G'$. Graph $G''$ is strictly-convex, since $G'$ is convex and is a subdivision of $G''$, and since all the internal vertices of $G''$ have degree greater than two. Moreover, since  $(G',\vec{d_s},L'_s,\gamma'_s)$ and  $(G',\vec{d_t},L'_t,\gamma'_t)$ are hierarchical-st convex graphs, it follows that $(G'',\vec{d_s},L''_s,\gamma''_s)$ and $(G'',\vec{d_t},L''_t,\gamma''_t)$ are hierarchical-st strictly-convex graphs, where $L''_s=L'_s\setminus\{\gamma'_s(u)\}$, $L''_t=L'_t\setminus\{\gamma'_t(u)\}$, $\gamma''_s(z)=\gamma'_s(z)$ for each vertex $z$ in $G''$, and $\gamma''_t(z)=\gamma'_t(z)$ for each vertex $z$ in $G''$. By Theorem~\ref{th:hong-nagamochi-revised}, $(G'',\vec{d_s},L''_s,\gamma''_s)$ and $(G'',\vec{d_t},L''_t,\gamma''_t)$ admit strictly-convex level planar drawings $\Lambda''_s$ and $\Lambda''_t$, respectively. 

\begin{figure}[t]
	\begin{center}
		\begin{tabular}{c c c c}
			\mbox{\includegraphics[scale=0.6]{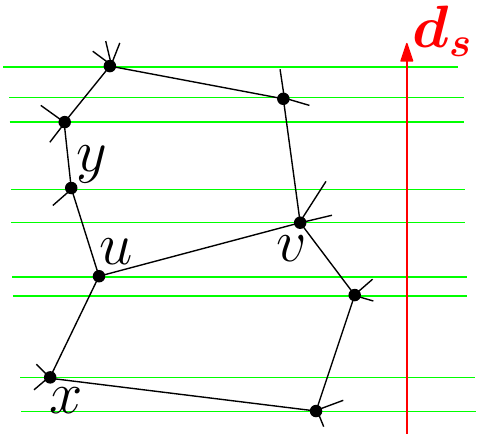}} \hspace{2mm} &
			\mbox{\includegraphics[scale=0.6]{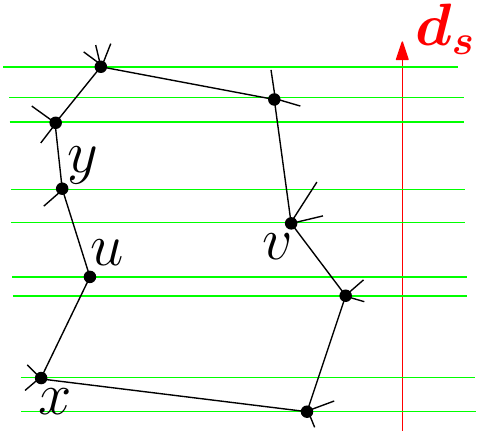}} \hspace{2mm} &
			\mbox{\includegraphics[scale=0.6]{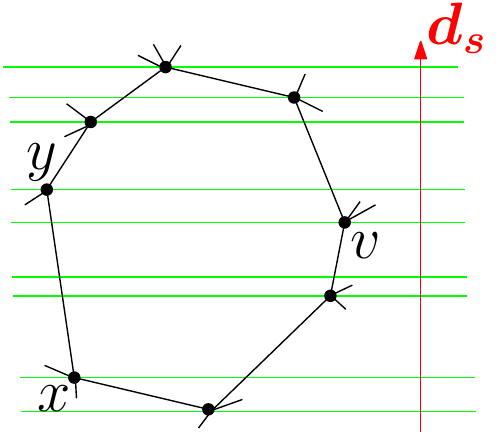}} \hspace{2mm} &
			\mbox{\includegraphics[scale=0.6]{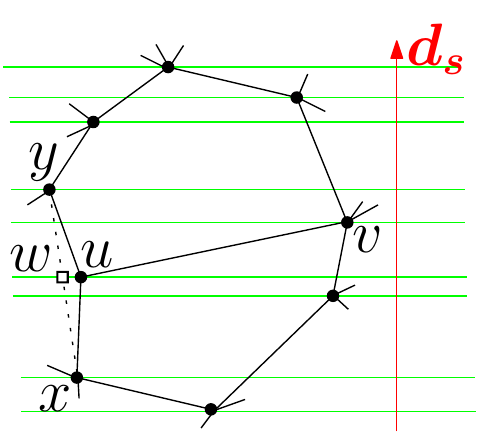}} \\
			(a) \hspace{2mm} & (b) \hspace{2mm} & (c) \hspace{2mm} & (d)
		\end{tabular}
		\caption{Drawings (a) $\Gamma_s$, (b) $\Lambda_s$, (c) $\Lambda''_s$, and (d) $\Gamma'_s$.}
		\label{fig:case2}
	\end{center}
\end{figure}

We modify $\Lambda''_s$ into a drawing $\Gamma'_s$ of $(G,\vec{d_s},L'_s,\gamma'_s)$, as in Fig.~\ref{fig:case2}. Assume w.l.o.g. that $\gamma'_s(x)<\gamma'_s(u)<\gamma'_s(y)$. Let $w$ be the intersection point of $\gamma'_s(u)$ and $\overline{xy}$ in $\Lambda''_s$ (where line $\gamma'_s(u)$ is the same as in $\Lambda_s$). Let $C''_{uv}$ be the facial cycle of $G''$ such that the facial cycle $C_{uv}$ of $G'$ is a subdivision of $C''_{uv}$. Insert $u$ in the interior of $C''_{uv}$, on $\gamma'_s(u)$, at distance $\epsilon>0$ from $w$. Remove edge $(x,y)$ from $\Lambda''_s$ and insert edges $(u,v)$, $(u,x)$, and $(u,y)$ as straight-line segments. Denote by $\Gamma'_s$ the resulting drawing. We have the following.

\begin{claimx} \label{cl:case-2}
	$\Gamma'_s$ is a strictly-convex level planar drawing of $(G,\vec{d_s},L'_s,\gamma'_s)$, provided that $\epsilon>0$ is sufficiently small.
\end{claimx}

\begin{proof}
	Drawing $\Gamma'_s$ is a level drawing of $(G,\vec{d_s},L'_s,\gamma'_s)$ since $\Lambda''_s$ is a level drawing of $(G'',\vec{d_s},L''_s,\gamma''_s)$ and since $u$ is on $\gamma'_s(u)$. Further, $\Gamma'_s$ is planar since $\Lambda''_s$ is planar and since edges $(u,v)$, $(u,x)$, and $(u,y)$ connect a point in the interior of the strictly-convex polygon delimiting a face in $\Lambda''_s$ with three points on the boundary of the same polygon. Finally, we prove that $\Gamma'_s$ is strictly-convex, provided that $\epsilon>0$ is sufficiently small. First, the outer face of $G$ in $\Gamma'_s$ is delimited by the same strictly-convex polygon as the outer face of $G''$ in $\Lambda''_s$. Moreover, consider any internal angle $\angle{azb}$ in $\Gamma'_s$ (in a straight-line planar drawing we call {\em internal angle} any angle internal to a polygon delimiting an internal face of the graph). If $\angle{azb}$ is also an internal angle in $\Lambda''_s$, then $\angle{azb}$ is convex in $f_z$, given that $\Lambda''_s$ is strictly-convex. Otherwise, $z$ is one of $u$, $v$, $x$, and $y$. If $z=v$, then $\angle{azb}$ is part of an internal angle in $\Lambda''_s$, hence it is convex in $f_z$. If $z=x$ (the case in which $z=y$ is analogous), then for any $\delta>0$, there exists an $\epsilon>0$ such that $\angle{azb}$ is at most $\delta$ radians plus an internal angle in $\Lambda''_s$, hence $\angle{azb}$ is convex in $f_z$, provided that $\epsilon$ is sufficiently small. Finally, all the angles incident to $u$ are convex, given that $u$ is internal to the triangle with vertices $x$, $y$, and $v$; hence $\angle{azb}$ is convex if $z=u$. 
\end{proof}

A strictly-convex level planar drawing $\Gamma'_t$ of $(G,\vec{d_t},L'_t,\gamma'_t)$ can be constructed analogously from $\Lambda''_t$. Drawings $\Gamma_s$ and $\Gamma'_s$ ($\Gamma_t$ and $\Gamma'_t$) are left-to-right equivalent, which can be proved as in Case~1. By Lemma~\ref{le:morphing-equivalent}, morphs \morph{\Gamma_s,\Gamma'_s} and \morph{\Gamma_t,\Gamma'_t} are strictly-convex and unidirectional. 

Inductively construct a unidirectional strictly-convex morph \morph{\Lambda''_s=\Lambda_0,\Lambda_1,\dots,\Lambda_\ell=\Lambda''_t} with $\ell\leq 2(n-1)+2(m-1)$ morphing steps  between $\Lambda''_s$ and $\Lambda''_t$. Let $0<\xi<1$ be sufficiently small so that the following holds true: For every $0\leq j\leq \ell$, insert $u$ in $\Lambda_j$ at a point which is a convex combination of the positions of $x$, $y$, and $v$ with coefficients $(\frac{1-\xi}{2},\frac{1-\xi}{2},\xi)$, remove edge $(x,y)$, and insert edges $(u,x)$, $(u,y)$, and $(u,v)$ as straight-line segments; then the resulting drawing $\Gamma_j$ of $G$ is strictly-convex. Such a $\xi>0$ exists. Namely, placing $v$ as a convex combination of the positions of $x$, $y$, and $v$ results in angles incident to $u$ and $v$ that are all convex. Moreover, as $\xi\to 0$, the point at which $u$ is placed approaches segment $\overline{xy}$, hence the size of any angle incident to $x$ or $y$ approaches the size of an angle incident to $x$ or $y$ in $\Lambda_j$, and the latter is strictly less than $\pi$ radians. 

With a single unidirectional strictly-convex linear morph, move $u$ in $\Gamma'_s$ to the point that is a convex combination of the positions of $x$, $y$, and $v$ with coefficients $(\frac{1-\xi}{2},\frac{1-\xi}{2},\xi)$; denote by $\Gamma''_s$ the drawing of $G$ obtained from this morph. Analogously, let \morph{\Gamma'_t,\Gamma''_t} be a  unidirectional strictly-convex linear morph, where the point at which $u$ is placed in $\Gamma''_t$ is a convex combination of the positions of $x$, $y$, and $v$ with coefficients $(\frac{1-\xi}{2},\frac{1-\xi}{2},\xi)$. 

For each $0\leq j\leq \ell -1$, $\Gamma_j$ and $\Gamma_{j+1}$ are left-to-right equivalent strictly-convex level planar drawings of the hierarchical-st strictly-convex graph $(G,\vec{d_j},L_j,\gamma_j)$, where $\vec{d_j}$ is an oriented straight line orthogonal to the direction of morph \morph{\Lambda_j,\Lambda_{j+1}}, $L_j$ is the set of lines through vertices of $G$ orthogonal to $\vec{d_j}$, and $\gamma_j$ maps each vertex of $G$ to the line in $L_j$ through it. In particular, $\Gamma_j$ and $\Gamma_{j+1}$ are strictly-convex drawings of $G$ since $\Lambda_j$ and $\Lambda_{j+1}$ are strictly-convex drawings of $G''$ and by the choice of $\xi$; further, every face of $G$ is an st-face in $\Gamma_j$ and $\Gamma_{j+1}$ by Lemmata~\ref{le:convex-is-monotone} and~\ref{le:two-polygons-monotone}; moreover, $u$ moves parallel to the other vertices since \morph{\Lambda_j,\Lambda_{j+1}} is unidirectional and since the points at which $u$ is placed in $\Gamma_j$ and $\Gamma_{j+1}$ are convex combinations of the positions of $x$, $y$, and $v$ with the same coefficients. By Lemma~\ref{le:morphing-equivalent}, \morph{\Gamma_j,\Gamma_{j+1}} is strictly-convex and unidirectional. Hence, \morph{\Gamma_s,\Gamma'_s,\Gamma''_s,=\Gamma_0,\Gamma_1,\dots,\Gamma_\ell=\Gamma''_t,\Gamma'_t,\Gamma_t} is a unidirectional strictly-convex morph between $\Gamma_s$ and $\Gamma_t$ with $\ell+4\leq 2n+2m$ morphing steps. 


{\bf Case~3} is very similar to Case~2, hence we only sketch the algorithm here. Let $G''$ be the graph obtained from $G'$ by replacing paths $(x_u,u,y_u)$ and $(x_v,v,y_v)$ with edges $(x_u,y_u)$ and $(x_v,y_v)$, respectively, where $x_u$ and $y_u$ ($x_v$ and $y_v$) are the only neighbors of $u$ (resp.\ $v$) in $G'$; $(G'',\vec{d_s},L''_s,\gamma''_s)$ and $(G'',\vec{d_t},L''_t,\gamma''_t)$ are hierarchical-st strictly-convex graphs, where $L''_s=L'_s\setminus\{\gamma'_s(u),\gamma'_s(v)\}$, $L''_t=L'_t\setminus\{\gamma'_t(u),\gamma'_t(v)\}$, $\gamma''_s(z)=\gamma'_s(z)$ for each vertex $z$ in $G''$, and $\gamma''_t(z)=\gamma'_t(z)$ for each vertex $z$ in $G''$. By Theorem~\ref{th:hong-nagamochi-revised}, $(G'',\vec{d_s},L''_s,\gamma''_s)$ and $(G'',\vec{d_t},L''_t,\gamma''_t)$ admit strictly-convex level planar drawings $\Lambda''_s$ and $\Lambda''_t$, respectively. 
We modify $\Lambda''_s$ into a strictly-convex level planar drawing $\Gamma'_s$ of $(G,\vec{d_s},L'_s,\gamma'_s)$ by inserting $u$ ($v$) on $\gamma'_s(u)$ (resp.\ $\gamma'_s(v)$) at distance $\epsilon>0$ from the intersection point of $\gamma'_s(u)$ with segment $\overline{x_uy_u}$ (of $\gamma'_s(v)$ with segment $\overline{x_vy_v}$) in the interior of the facial cycle of $G''$ such that the facial cycle $C_{uv}$ of $G'$ is a subdivision of $C''_{uv}$. Analogously, we modify $\Lambda''_t$ into a strictly-convex level planar drawing $\Gamma'_t$ of $(G,\vec{d_t},L'_t,\gamma'_t)$. Drawings $\Gamma_s$ and $\Gamma'_s$ ($\Gamma_t$ and $\Gamma'_t$) are left-to-right equivalent. 

Inductively construct a unidirectional strictly-convex morph \morph{\Lambda''_s=\Lambda_0,\dots,\Lambda_\ell=\Lambda''_t} with $\ell\leq 2(n-2)+2(m-1)$ morphing steps. Let $\xi>0$ be sufficiently small so that for every $0\leq j\leq \ell$, inserting $u$ ($v$) in $\Lambda_j$ at a convex combination of the positions of $x_u$, $y_u$, $x_v$, and $y_v$ with coefficients $(\frac{1-\xi}{2},\frac{1-\xi}{2},\frac{\xi}{2},\frac{\xi}{2})$ (resp.\ $(\frac{\xi}{2},\frac{\xi}{2},\frac{1-\xi}{2},\frac{1-\xi}{2})$), removing edges $(x_u,y_u)$ and $(x_v,y_v)$, and inserting edges $(x_u,u)$, $(y_u,u)$, $(x_v,v)$, $(y_v,v)$, and $(u,v)$ results in a strictly-convex drawing $\Gamma_j$ of $G$. With a unidirectional strictly-convex linear morph \morph{\Gamma'_s,\Gamma''_s}, move $u$ in $\Gamma'_s$ to the point that is a convex combination of the positions of $x_u$, $y_u$, $x_v$, and $y_v$ with coefficients $(\frac{1-\xi}{2},\frac{1-\xi}{2},\frac{\xi}{2},\frac{\xi}{2})$. With a unidirectional strictly-convex linear morph \morph{\Gamma''_s,\Gamma'''_s}, move $v$ in $\Gamma''_s$ to the point that is a convex combination of the positions of $x_u$, $y_u$, $x_v$, and $y_v$ with coefficients $(\frac{\xi}{2},\frac{\xi}{2},\frac{1-\xi}{2},\frac{1-\xi}{2})$. Define morph \morph{\Gamma'_t,\Gamma''_t,\Gamma'''_t} analogously. For each $0\leq j\leq \ell -1$, $\Gamma_j$ and $\Gamma_{j+1}$ are left-to-right equivalent strictly-convex level planar drawings of the hierarchical-st strictly-convex graph $(G,\vec{d_j},L_j,\gamma_j)$, where $\vec{d_j}$ is an oriented line orthogonal to the direction of morph \morph{\Lambda_j,\Lambda_{j+1}}, $L_j$ is the set of lines through vertices of $G$ and orthogonal to $\vec{d_j}$, and $\gamma_j$ maps each vertex of $G$ to the line in $L_j$ through it. By Lemma~\ref{le:morphing-equivalent}, \morph{\Gamma_s,\Gamma'_s,\Gamma''_s,\Gamma'''_s=\Gamma_0,\dots,\Gamma_\ell=\Gamma'''_t,\Gamma''_t,\Gamma'_t,\Gamma_t} is a unidirectional strictly-convex morph between $\Gamma_s$ and $\Gamma_t$ with $\ell +6\leq 2n+2m$ morphing steps. 
We get the following.


\begin{theorem}\label{th:strictly-main}
There exists an algorithm to construct a strictly-convex unidirectional morph  with $O(n)$ morphing steps between any two strictly-convex drawings of the same $n$-vertex plane graph.
\end{theorem}


A simple enhancement of the above described algorithm allows us to extend our results to (non-strictly) convex drawings of convex graphs. We have the following.

\begin{theorem}\label{th:main}
There exists an algorithm to construct a convex unidirectional morph with $O(n)$ morphing steps between any two convex drawings of the same $n$-vertex plane graph.
\end{theorem}


\begin{proof}
Consider drawing $\Gamma_s$ and let $P=(u_1,\dots,u_k)$ be any maximal path in the cycle $C$ delimiting the outer face $f$ of $G$ such that $u_2,\dots,u_{k-1}$ are degree-$2$ vertices of $G$ that are flat in $f$. Let $a$ be a circular arc between $u_1$ and $u_k$ that is monotone with respect to the direction of $\overline{u_1u_k}$ and that forms a convex curve with $C-\{u_2,\dots,u_{k-1}\}$. Move $u_2,\dots,u_{k-1}$ on $a$ with one morphing step in the direction orthogonal to $\overline{u_1u_k}$. Repeating this operation for every path $P$ satisfying the above properties results in a convex unidirectional morph with $O(n)$ morphing steps between $\Gamma_s$ and a convex drawing $\Gamma'_s$ of $G$ such that the polygon delimiting the outer face of $G$ is strictly-convex. Apply the same algorithm to construct a convex unidirectional morph with $O(n)$ morphing steps between $\Gamma_t$ and a convex drawing $\Gamma'_t$ of $G$ such that the polygon delimiting the outer face of $G$ is strictly-convex.

Consider any maximal path $P=(u_1,\dots,u_k)$ such that $u_2,\dots,u_{k-1}$ are degree-$2$ internal vertices of $G$. For $2\leq i\leq k-1$, the position of $u_i$ is a convex combination of the positions of $u_1$ and $u_k$ both in $\Gamma'_s$ and in $\Gamma'_t$ (the coefficients of such convex combinations are, in general, different in $\Gamma'_s$ and in $\Gamma'_t$). With a single linear morph in the direction of $\overline{u_1u_k}$, move each of $u_2,\dots,u_{k-1}$ in $\Gamma'_s$ to the point which is a convex combination of the positions of $u_1$ and $u_k$ with the same coefficients as in $\Gamma'_t$. Repeating this operation for every path $P$ satisfying the above properties results in a convex unidirectional morph with $O(n)$ morphing steps between $\Gamma'_s$ and a convex drawing $\Gamma''_s$ of $G$ such that the polygon delimiting the outer face of $G$ is strictly-convex and such that, for each maximal path $(u_1,\dots,u_k)$ where $u_2,\dots,u_{k-1}$ are degree-$2$ internal vertices of $G$, the coefficients of $u_i$ as a convex combination of $u_1$ and $u_k$ are the same in $\Gamma''_s$ and in $\Gamma'_t$.

Replace each maximal path $(u_1,\dots,u_k)$ in $G$ such that $u_2,\dots,u_{k-1}$ are degree-$2$ internal vertices of $G$ with an edge $(u_1,u_k)$. Denote by $G'$ the resulting graph; by Theorems~\ref{th:convex-characterization} and~\ref{th:strictly-convex-characterization}, $G'$ is strictly-convex. Denote by $\Lambda''_s$ and $\Lambda'_t$ the drawings of $G'$ obtained respectively from $\Gamma''_s$ and $\Gamma'_t$ by replacing each path $(u_1,\dots,u_k)$ as above with an edge $(u_1,u_k)$. Compute a strictly-convex unidirectional morph \morph{\Lambda''_s=\Lambda_0,\Lambda_1,\dots,\Lambda_\ell=\Lambda'_t} with $\ell\in O(n)$ morphing steps as in Theorem~\ref{th:strictly-main}. For each path $(u_1,\dots,u_k)$ satisfying the above properties, for each $2\leq i\leq k-1$ and $1\leq j\leq \ell-1$, draw $u_i$ in $\Lambda_j$ at a point that is the convex combination of the positions of $u_1$ and $u_k$ in $\Lambda_j$ with the same coefficients as in $\Gamma''_s$ and in $\Gamma'_t$. This results in a morph \morph{\Gamma''_s=\Gamma_0,\Gamma_1,\dots,\Gamma_\ell=\Gamma'_t}, which is convex and unidirectional. Namely, in every morphing step \morph{\Gamma_j,\Gamma_{j+1}}, vertex $u_i$ moves between two points that are the convex combinations of the positions of $u_1$ and $u_k$ with the same coefficients, hence it moves parallel to each of $u_1$ and $u_k$ (from which \morph{\Gamma''_s=\Gamma_0,\Gamma_1,\dots,\Gamma_\ell=\Gamma'_t} is unidirectional) and it stays on $\overline{u_1u_k}$ at any time instant of \morph{\Gamma_j,\Gamma_{j+1}} (from which \morph{\Gamma''_s=\Gamma_0,\Gamma_1,\dots,\Gamma_\ell=\Gamma'_t} is convex). Hence, \morph{\Gamma_s,\dots,\Gamma'_s,\dots,\Gamma''_s=\Gamma_0,\Gamma_1,\dots,\Gamma_\ell=\Gamma'_t,\dots,\Gamma_t} is a unidirectional convex morph between $\Gamma_s$ and $\Gamma_t$ with $O(n)$ morphing steps. 
\end{proof}






\bibliography{bibliography}

\begin{thebibliography}{10}

\bibitem{aac-mpgdpns-13}
S.~Alamdari, P.~Angelini, T.~M. Chan, G.~{Di Battista}, F.~Frati, A.~Lubiw,
  M.~Patrignani, V.~Roselli, S.~Singla, and B.~T. Wilkinson.
\newblock Morphing planar graph drawings with a polynomial number of steps.
\newblock In S.~Khanna, editor, {\em SODA}, pages 1656--1667, 2013.

\bibitem{abcdfm-ccapdg-15}
G.\ Aloupis, L.\ Barba, P.\ Carmi, V.\ Dujmovic, F.\ Frati, and P.\ Morin.
\newblock Compatible connectivity-augmentation of planar disconnected graphs.
\newblock In P.~Indyk, editor, {\em SODA}, pages 1602--1615, 2015.

\bibitem{addfpr-mpgdo-14}
P.~Angelini, G.~{Da Lozzo}, G.~{Di Battista}, F.~Frati, M.~Patrignani, and
  V.~Roselli.
\newblock Morphing planar graph drawings optimally.
\newblock In J.~Esparza, P.~Fraigniaud, T.~Husfeldt, and E.~Koutsoupias,
  editors, {\em {ICALP}}, volume 8572 of {\em LNCS}, pages 126--137, 2014.

\bibitem{afpr-mpgde-13}
P.~Angelini, F.~Frati, M.~Patrignani, and V.~Roselli.
\newblock Morphing planar graph drawings efficiently.
\newblock In S.~Wismath and A.~Wolff, editors, {\em GD}, volume 8242 of {\em
  LNCS}, pages 49--60, 2013.

\bibitem{br-csc3c-06}
I.~B{\'{a}}r{\'{a}}ny and G.~Rote.
\newblock Strictly convex drawings of planar graphs.
\newblock {\em Documenta Mathematica}, 11:369--391, 2006.

\bibitem{bg-stcc3p-69}
D.~Barnette and B.~Gr{\"{u}}nbaum.
\newblock On {S}teinitz's theorem concerning convex 3-polytopes and on some
  properties of planar graphs.
\newblock In {\em Many Facets of Graph Theory}, volume 110 of {\em Lecture
  Notes in Mathematics}, pages 27--40. Springer, 1969.

\bibitem{bhl-mpgdum-13}
F.~{Barrera-Cruz}, P.~Haxell, and A.~Lubiw.
\newblock Morphing planar graph drawings with unidirectional moves.
\newblock Mexican Conference on Discr. Math. and Comput. Geom., 2013.

\bibitem{bfm-cdcpg-07}
N.~Bonichon, S.~Felsner, and M.~Mosbah.
\newblock Convex drawings of 3-connected plane graphs.
\newblock {\em Algorithmica}, 47(4):399--420, 2007.

\bibitem{c-dprc-44}
S.~Cairns.
\newblock Deformations of plane rectilinear complexes.
\newblock {\em Am. Math. Mon.}, 51:247--252, 1944.

\bibitem{cyn-lacdp-84}
N.~Chiba, T.~Yamanouchi, and T.~Nishizeki.
\newblock Linear algorithms for convex drawings of planar graphs.
\newblock In J.~A. Bondy and U.~S.~R. Murty, editors, {\em Progress in Graph
  Theory}, pages 153--173. Academic Press, New York, NY, 1984.

\bibitem{ck-cgd3pg-97}
M.~Chrobak and G.~Kant.
\newblock Convex grid drawings of 3-connected planar graphs.
\newblock {\em Int. J. Comput. Geometry Appl.}, 7(3):211--223, 1997.

\bibitem{ekp-ifmpg-03}
C.~Erten, S.~G. Kobourov, and C.~Pitta.
\newblock Intersection-free morphing of planar graphs.
\newblock In G.~Liotta, editor, {\em GD}, volume 2912 of {\em LNCS}, pages
  320--331, 2004.

\bibitem{fe-gdm-02}
C.~Friedrich and P.~Eades.
\newblock Graph drawing in motion.
\newblock {\em J.\ Graph Alg.\ Ap.}, 6:353--370, 2002.

\bibitem{gs-gifpm-01}
C.~Gotsman and V.~Surazhsky.
\newblock Guaranteed intersection-free polygon morphing.
\newblock {\em Computers {\&} Graphics}, 25(1):67--75, 2001.

\bibitem{gs-tgopg-81}
B.~Grunbaum and G.C. Shephard.
\newblock {\em The geometry of planar graphs}.
\newblock Camb. Univ. Pr., 1981.

\bibitem{hn-cdhpgcpg-10}
S.~H. Hong and H.~Nagamochi.
\newblock Convex drawings of hierarchical planar graphs and clustered planar
  graphs.
\newblock {\em J. Discrete Algorithms}, 8(3):282--295, 2010.

\bibitem{hn-ltascditpg-10}
S.~H. Hong and H.~Nagamochi.
\newblock A linear-time algorithm for symmetric convex drawings of internally
  triconnected plane graphs.
\newblock {\em Algorithmica}, 58(2):433--460, 2010.

\bibitem{rnn-rgdpg-98}
M.~S. Rahman, S.~I. Nakano, and T.~Nishizeki.
\newblock Rectangular grid drawings of plane graphs.
\newblock {\em Comput. Geom.}, 10(3):203--220, 1998.

\bibitem{rng-rdpg-04}
M.~S. Rahman, T.~Nishizeki, and S.~Ghosh.
\newblock Rectangular drawings of planar graphs.
\newblock {\em J. of Algorithms}, 50:62--78, 2004.

\bibitem{s-rc3clt-13}
J.~M. Schmidt.
\newblock Contractions, removals, and certifying 3-connectivity in linear time.
\newblock {\em {SIAM} J. Comput.}, 42(2):494--535, 2013.

\bibitem{sg-cmcpt-01}
V.~Surazhsky and C.~Gotsman.
\newblock Controllable morphing of compatible planar triangulations.
\newblock {\em ACM Trans. Graph}, 20(4):203--231, 2001.

\bibitem{sg-imct-03}
V.~Surazhsky and C.~Gotsman.
\newblock Intrinsic morphing of compatible triangulations.
\newblock {\em Internat. J. of Shape Model.}, 9:191--201, 2003.

\bibitem{t-pdfig-80}
C.~Thomassen.
\newblock Planarity and duality of finite and infinite graphs.
\newblock {\em J. Comb. Theory, Ser. {B}}, 29(2):244--271, 1980.

\bibitem{t-dpg-83}
C.~Thomassen.
\newblock Deformations of plane graphs.
\newblock {\em J. Comb. Th. Ser. B}, 34(3):244--257, 1983.

\bibitem{t-prg-84}
C.~Thomassen.
\newblock Plane representations of graphs.
\newblock In J.~A. Bondy and U.~S.~R. Murty, editors, {\em Progress in Graph
  Theory}, pages 43--69. Academic Press, New York, NY, 1984.

\end{thebibliography}
\bibliographystyle{plain}

\end{document}